\newif\ifjpc
\jpcfalse 
\ifjpc
\documentclass{jpc}
\else
\documentclass[USenglish,oneside]{article}
\usepackage[letterpaper, portrait, margin=1in]{geometry}
\fi

\ifjpc
\fi

\usepackage{natbib}
\usepackage[ruled]{algorithm2e}
\ifjpc
\theoremstyle{plain}
\else
\usepackage{amsthm}
\newtheorem{thm}{Theorem}
\newtheorem{cor}{Corollary}
\newtheorem{rem}{Remark}
\theoremstyle{definition}
\newtheorem{defi}{Definition}
\usepackage{xcolor}
\definecolor{citegreen}{HTML}{458B00}
\usepackage{hyperref}
\hypersetup{
   colorlinks=true,
   citecolor=citegreen
}
\fi

\usepackage{comment}
\usepackage{amsmath}
\usepackage{amssymb}
\usepackage{stmaryrd}
\usepackage{setspace}
\usepackage{multirow}
\usepackage{tabularx,ragged2e,booktabs,caption}
\usepackage{color}
\usepackage[thinlines]{easytable}
\usepackage{array}
\usepackage{lipsum}
\usepackage{multirow}
\usepackage{textcomp}
\usepackage{subcaption}
\usepackage{caption}
\usepackage{graphicx}
\usepackage{tikz}
\usetikzlibrary{decorations.pathmorphing}

\newcommand{\descr}[1]{\vspace{0.2cm} \noindent \textbf{#1}}

\begin{document}

\ifjpc
\else
\title{Sharing in a Trustless World: Privacy-Preserving Data Analytics with Potentially Cheating Participants}
\author{Tham Nguyen$^{1}$, Hassan Jameel Asghar$^{1,2}$, Raghav Bhakar$^{2}$, \\
Dali Kaafar$^{1}$ \& Farhad Farokhi$^{3}$\\
\small $^{1}$Macquarie University, Australia\\
\small \texttt{\{tham.ngyuen, hassan.asghar, dali.kaafar\}@mq.edu.au}\\
\small $^{2}$Data61, CSIRO\\
\small \texttt{raghav.bhaskar@csiro.au}\\
\small $^{3}$University of Melbourne\\
\small \texttt{farhad.farokhi@unimelb.edu.au}
}

\maketitle

\fi

\begin{abstract}
\noindent Lack of trust between organisations and privacy concerns about their data are impediments to an otherwise potentially symbiotic joint data analysis.
We propose DataRing, a data sharing system that allows mutually mistrusting participants to query each others' datasets in a privacy-preserving manner while ensuring the correctness of input datasets and query answers even in the presence of (cheating) participants deviating from their true datasets. By relying on the assumption that if only a small subset of rows of the true dataset are known, participants cannot submit answers to queries deviating significantly from their true datasets. 
We employ differential privacy and a suite of cryptographic tools to ensure individual privacy for each participant's dataset and data confidentiality from the system. Our results show that the evaluation of 10 queries on a dataset with 10 attributes and 500,000 records is achieved in 90.63 seconds. 
DataRing could detect cheating participant that deviates from its true dataset in few queries with high accuracy.
\end{abstract}

\section{Introduction}
\label{sec:intro}
Joint analysis on multiple datasets owned by different organizations (parties) has the potential to unlock numerous benefits to the participating organizations as well as the society in general. Example areas include medical research, financial fraud detection, and international cyber defence~\citep{Froelicher2017}. While this incentivizes organizations to share data, a major concern is maintaining the privacy of individuals who contribute their data to an organization's dataset. A further concern is if the organizations are mutually distrusting; a dishonest party may not share its ``true'' dataset with the self-serving aim of protecting its business competitive advantage and maximizing its own utility. Although this can be approached using legislative frameworks, whereby parties are legally forced to share their true data, our interest is in a technical solution to detect dishonest behaviour while minimizing the reliance on legal solutions which might deter and/or slow down the adoption of data sharing and joint analysis platforms.

Many systems to share data in a privacy-preserving manner have been proposed in the research literature, built on a combination of cryptographic protocols~\citep{Froelicher2017,Froelicher2019}, differentially private mechanisms~\citep{Froelicher2017, hynes-tee}, and trusted execution environments~\citep{hynes-tee, ohrimenko2016oblivious, hunt2018chiron}. An underlying trust assumption in most of these proposals is that the participants do not deviate from their original datasets or the analysis on their datasets reflects the true result. Reliance on trust is not surprising, since ensuring correctness of inputs is a difficult problem in general. Indeed, in general cryptographic protocols secure against active adversaries, it is assumed that the malicious behaviour (from active adversaries) includes deviating from the true input, since the protocol cannot determine if the claimed input is true or not~\citep[\S 7.2.3]{foc-2}. 

A few recent works in privacy-preserving data sharing have circumnavigated this hurdle by checking if the input satisfies some publicly known relation~\citep{Froelicher2019, corrigan2017prio}, e.g., age being within 0 and 150 years. However, the participants may still deviate from their data as long as it remains within these semantic bounds. Helen~\citep{helen2019} goes further, ensuring consistently of data analysis by relying on encrypted summaries (commitment) of each participant's dataset. The trust assumption, however, is that the participant constructs its summary faithfully using its true dataset in entirety. The central aim of this paper is to construct a system that provides (approximately) accurate data analysis while minimising the trust assumption on the true input dataset, i.e., the fraction of the dataset assumed to be true is as small as possible. 

We propose DataRing, a system that allows mutually mistrusting participants to query each others' datasets in a privacy-preserving manner while ensuring the correctness of input datasets and query answers. An adversary (cheating participant) may modify its dataset to provide incorrect answers to queries. The goal of the adversary is to abuse the data sharing process benefiting from other participants' data without contributing its own. Our salient contribution is the methodology to ensure correctness of input data and query answers. This is based on the observation that the vast majority of records in high-dimensional real-world datasets are unique~\citep[\S 4]{Vincent2017},\footnote{Also see Section~\ref{sec:exp_setup}.}
and if a small random number of such unique records are known by the system\footnote{See Remark~\ref{rem:size_of_backgroundKnow}.}, then the participants, not knowing the exact identities of these records, can only marginally deviate from their true datasets. In more detail, the participant first uploads a random sample of its dataset, which we call the \emph{partial view}. Using a small set of records from the true dataset, the system can verify the correctness of the partial view by relying on the probabilistic properties of the random sample. This ensures that the partial view is as close to a random sample of the true dataset as possible. In a subsequent query phase, the correctness of the answers to queries on the participant's true dataset, is ensured by using hidden \emph{test} queries based on the aforementioned small set of records, verified partial view and meta data, e.g., dataset size and partial view size. 

We call the small sample of random records \textit{background knowledge}, which is assumed to be sampled from the participant's true dataset and prior to the initialization of the data sharing. This background knowledge can be obtained either technically, using for instance, a version of the partial view protocol proposed in this paper (cf. Section \ref{sec:partial_view}), or under legislative requirements. An example of the latter is national financial intelligence agencies requiring records of financial transactions for auditing. As mentioned above, the advantage of DataRing, unlike Helen~\citep{helen2019}, is that this assumes only a small fraction of the entire dataset to be true (cf. Remark \ref{rem:size_of_backgroundKnow}).

Apart from the correctness of input datasets and query answers, the DataRing system also ensures confidentiality of participants' data, query's content and query answer, and privacy of individuals in each participant's dataset. To ensure the confidentiality of each participant's dataset, all computations on participants' data are done in the encrypted domain. 
This ensures that the servers do not learn the contents of the input dataset, the partial view, the queries and their answers. Additionally, the participants (apart from knowing their own datasets) only learn answers to their queries on another participant's dataset. To accomplish this, we employ a suite of cryptographic tools including an additive homomorphic encryption scheme~\citep{koblitz1987elliptic}, collective public key of servers~\citep{Froelicher2017,Froelicher2019}, threshold decryption~\citep{Froelicher2017}, and re-encryption~\citep{Froelicher2017}, and provide formal proofs of security for our confidentiality claims. Furthermore, to ensure the privacy of individuals in a participant's data, we employ differentially private noise~\citep{calib-noise} to the query answers. 
Our schemes might be reminiscent of Private Data Retrieval (PIR) technique~\citep{pir_sergey,SPIR,xpir}. Nonetheless, PIR techniques are not applicable in our setting because PIR aims to hide the content of the query from the database owner and does not ensure the database confidentiality. Existing PIR schemes are implemented in such a way that either sending the whole database to the querying party or using  mutually-distrustful replicated databases at multiple servers.\footnote{Also see Remark~\ref{rem:pir_pv},~\ref{rem:query_compute}.
}

Finally, we thoroughly evaluate the robustness of DataRing in detecting cheating participants and its performance in terms of computational and communication overhead, by implementing an API in C/C++. Our experimental evaluation shows that DataRing enforces a participant to provide a partial view constructed from a dataset close to its true dataset. Particularly, a participant with a dataset of 500,000 records must use at least 97.15\% of its true dataset to generate its partial view in order to be permitted to join the data sharing (i.e. passing the partial view verification with a probability of 0.95). In addition, DataRing can detect a cheating participant with high accuracy. 
On the performance side, (encrypted) evaluation of 10 queries on a dataset with 10 attributes and 500,000 records is achieved in 90.63 seconds.

The rest of the paper is organized as follows. Section \ref{sec:background} covers preliminaries needed for the rest of the paper. We describe the DataRing system in \ref{sec:system_model}. Section \ref{sec:partial_view} gives detail of the partial view phase and Section \ref{sec:query_computation} covers the query evaluation phase. Section \ref{sec:exp_setup} contains setup for our evaluation. We present detailed security evaluation of DataRing in Section \ref{sec:evaluation} and its performance in Section \ref{sec:performance}. We present related work in Section \ref{sec:related_work} and our conclusion in Section \ref{sec:conclusion}.

\section{Preliminaries}
\label{sec:background}

\subsection{System Model}\label{sec:asump_notation}
The DataRing system model consists of a set of participants $P_1, P_2, \ldots$, where each participant $P_i$ owns a dataset $D_i$ containing $N_i$ records. A record is perceived as belonging to an individual, e.g., one of participant $P_i$'s customers. The participants wish to analyse each other's datasets in the form of queries. To facilitate this, the system consists two non-colluding servers $S_1, S_2$. Abusing notation, we shall use the term ``server'' to mean the set of all servers, and simply use $\mathcal{S}$ to denote it. In practice, this denotes the combined role played by the servers. A generic participant shall be denoted by $P$, with its dataset denoted $D$ containing $N$ records. Figure \ref{fig:system_model} illustrates the DataRing system.

\begin{figure}[t]
\centering
\includegraphics[width=0.7\linewidth]{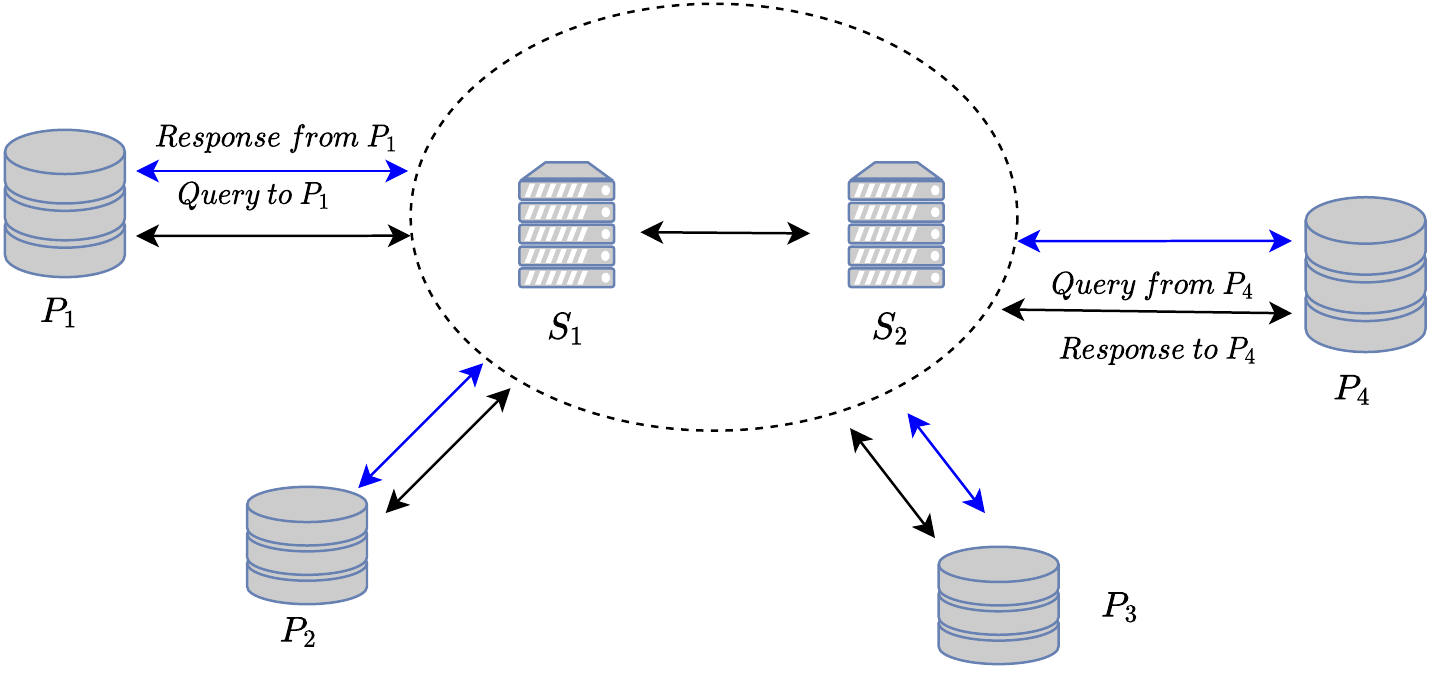}
\caption{DataRing with 4 participants $P_1$ to $P_4$ and 2 servers $S_1$, $S_2$, blue arrows indicate the data analysis of other participants on $P_i$'s dataset, black arrows indicate the data analysis of $P_i$ on other participants' dataset. \label{fig:system_model}
}
\end{figure}

\subsection{Threat model}\label{subsec:threat_model}
Our threat model is as follows. 
\begin{itemize}
    \item Each of the servers in $\mathcal{S}$ is honest-but-curious, i.e., each server performs the protocol steps faithfully but would like to infer information from the messages received. We also assume they do not collude. This means that $\mathcal{S}$ is honest as a whole.
    
    \item Each participant $P$ is potentially cheating, i.e., it can modify its true dataset $D$, and/or give arbitrary answer to any query on $D$. The exception is that $P$ honestly provides metadata including the domain (cf. Section~\ref{subsec:dataset}) and the size of its dataset.
    
    \item The participants do not collude with each other.
    
\end{itemize}

\subsection{Desired Properties}
\label{sub:desired}
DataRing is designed to satisfy the following properties.

\descr{Confidentiality.}
We have three main requirements

\begin{itemize}
    \item A participant's dataset is not disclosed to other participants and the server.
    \item The content of queries submitted by a participant to be evaluated on the dataset of another participant is not disclosed to other participants (including the one receiving the queries) and the server.
    \item Query responses can only be viewed by the inquiring participant.
\end{itemize}
The only exception to the above is background knowledge of the server about each participant's dataset, which we discuss in Section~\ref{subsec:assumptions}.

\descr{Individual Privacy.} Answers to queries should protect the privacy of individuals contributing their data to the dataset of a participant.

\descr{Correctness.} Query answers from participants should be as close to the true answers as possible.  

\subsection{Datasets and Queries}\label{subsec:dataset}
This section gives details of the dataset representation and query type.

\descr{Datasets.} We represent datasets as histograms over a public domain $\mathcal{D}$~\citep{Dwork2014}. More precisely, a dataset $D$ is an element of the set $\mathbb{N}^{|\mathcal{D}|}$. Each of the $i = 1, \ldots, |\mathcal{D}|$ members of the domain is called a \emph{record type} or \emph{label}. These are the possible types that a database can take as its records (with possibly multiple records of the same type). We assume that the labels in $\mathcal{D}$ are enumerated as $\{1, 2, \ldots, |\mathcal{D}|\}$ by a publicly known ordering, and hence we shall refer to a label simply by its index. It will often be convenient to denote the dataset as a set of label-value pairs $(i, \mathsf{val}(i))$, where $i \in \mathcal{D}$ is a label and $\mathsf{val}(i) \in \mathbb{N}$ is the number of times $i$ appears in the dataset $D$. The size $N$ of the dataset $D$ is the number of records in $D$ and is given by its $l_1$ norm, i.e., $N = \lVert D \rVert_1$. 

\begin{figure}
\begin{center}
    \begin{tabular}{c|c|c}
    \multicolumn{3}{c}{Dataset}\\
    \hline
    Gen. & Home & Loan\\
    \hline \hline
    F & Rent & 10K\\
    M & Own & 20K\\
    M & Rent & 10K\\
    F & Own & 20K \\
    \end{tabular}
    \quad
    \begin{tabular}{c|c|c|c|c|c|c|c|c}
    \multicolumn{9}{c}{Histogram Representation}\\
    \hline
    Label $i$ & $1$ & $2$ & $3$ & $4$ & $5$ & $6$ & $7$ & $8$\\
    \hline \hline
    $\mathsf{val}(i)$ & $1$ & $0$ & $0$ & $1$ & $1$ & $0$ & $0$ & $1$\\
    \end{tabular}
\end{center}
\caption{Example of a dataset and its histogram representation.} 
\label{fig:dataset_example}
\end{figure}
Figure \ref{fig:dataset_example} shows an example of a small dataset $D$ and its histogram representation. The domain consists of three attributes \textit{Gender, Home ownership, Loan amount}, which can take two possible values each, i.e., $\{Female, Male\}$, $\{Rent, Own\}$ and $\{10K,20K\}$, respectively. Thus, the domain has $2^3 = 8$ possible labels (record types). The dataset itself consists of four unique records. The histogram representation consists of all possible data points with labels, which are enumerated by some public ordering, e.g., $1 \leftrightarrow (F,Rent,10K), \ldots, 8 \leftrightarrow (M,Own,20K)$. The value of the label $i$, i.e.,  $\mathsf{val}(i)$ is $1$ if it is in the dataset, and $0$, otherwise.

\descr{Queries.} We restrict our focus to \emph{count} queries as these are powerful primitives for capturing many statistics from a database \citep[\S 3.3]{Dwork2014}. Namely, a query $Q$ on the data $D$ is defined as the sum of the result of a predicate on each row of $D$. We denote this by $Q(D)$. Notice that we can write query $Q$ as a binary vector $\in \{0, 1\}^{|\mathcal{D}|}$, and then $Q(D)$ as the dot product $\langle Q, D \rangle \in [0, N]$. 

\subsection{Building Blocks}\label{subsec:crypto}
For individual privacy, we employ differential privacy~\citep{calib-noise, Dwork2014} to respond to queries, and for confidentiality, we predominantly use a partially homomorphic encryption scheme. 

\subsubsection{Differential Privacy}
Here we highlight two important considerations:

\descr{Privacy Budget.} Each participant $P_i$ has a \emph{total} privacy budget (under differential privacy) for every other participant $P_j$, $j \neq i$, denoted $\epsilon_{i, j}$. Furthermore, participant $P_i$ also has a separate \emph{total} privacy budget for the server $\mathcal{S}$ to facilitate the query evaluation of $P_j$ on $P_i$'s dataset, denoted $\epsilon_{i, \mathcal{S}}$. For simplicity, we assume that $\epsilon_{i, j} = \epsilon_{i, \mathcal{S}}$ for all $j$. 

\descr{Individual's Record.} Our privacy guarantees for individuals in the dataset are tied to differential privacy. 
For simplicity, we assume that each record of the dataset belongs to a unique individual, i.e., each individual does not have more than one entry (e.g., transaction) in $D$. In case of multiple records from individuals, the privacy budget can be scaled accordingly. 

\subsubsection{Additive Homomorphic Encryption}\label{subsec:encryption}
We use the Elliptic Curve ElGamal (EC ElGamal) cryptosystem \citep{koblitz1987elliptic} \citep[\S 1.2.3]{guideToECC} as an additive homomorphic encryption scheme~\citep{phe_survey}. The scheme is described as follows. 

Let $\mathcal{E}$ be an elliptic curve over a finite field $\mathbb{F}_p$ and let $B$, a point on $\mathcal{E}(\mathbb{F}_p)$, be the generator of the cyclic subgroup of prime order $q$. Given a message $x$ from a plaintext space which is a subset of $[0, q-1]$, we define the mapping of $x$ to a point $X$ on the curve as $X = xB$. The reverse mapping, retrieving $x$ from $X = xB$ amounts to solving the discrete logarithm problem, and can be done efficiently, for a small message space via a lookup table \citep[\S 3.3]{Talos_ECEG}. EC ElGamal is semantically secure under the decisional Diffie-Hellman (DDH) assumption. 
Details of basic operations such as encryption, decryption, scalar multiplication of ciphertexts appear in Appendix~\ref{app:ec-elgamal}. Let $k$ denote the private key which is sampled uniformly at random from $[1, q-1]$. The public key is $K = kB$.

\descr{Additive homomorphic property.} \citep{Pilatus_ECEG,Talos_ECEG,Froelicher2017,Froelicher2019}
Given the EC ElGamal encryption of two messages $X_1 = x_1B$ and $X_2 = x_2B$ as $\mathsf{Enc}_K(x_1; r_1) = (r_1B, X_1 + r_1K)$ and $\mathsf{Enc}_K(x_2; r_2) = (r_2B, X_2 + r_2K)$, the addition of two ciphertexts is $\mathsf{Enc}(x_1 + x_2; r_1 + r_2) = ((r_1+r_2)B, X_1 + X_2 + (r_1 + r_2)K)$. The resulting ciphertext can be decrypted as $\mathsf{Dec}_k(\mathsf{Enc}_K(x_1 + x_2;r_1 + r_2)) = X_1 + X_2$, from which we can recover $x_1 + x_2$. 

\descr{Collective public key.}
We employ the concept of collective public key used in~\citep{Froelicher2017,Froelicher2019}. Namely, servers $\mathcal{S} = \{S_1, S_2\}$ having public-private key pairs $(K_i, k_i)$ sum their public key as $K_\mathcal{S} = \sum_{i=1}^{2}K_i$ to obtain the collective public key as a point on the curve, shared among all parties. Note that the private keys of each party are never shared with any other party. %

To decrypt a message $X$ (mapped on the curve) encrypted under the collective public key $K_\mathsf{S}$, all servers must participate. Notice that at no point in the decryption is the private key of any server shared with any other server. The $t$ servers can also collectively re-encrypt message $X$ to be under any public key $U$~\citep{Froelicher2017}. 

\descr{Threshold decryption.}~\citep{Froelicher2017}
Given the encryption of the message $X$ (mapped on the curve) as $(C^{(0)}_1, C^{(0)}_2) = (rB, X+rK_S)$ via the collective key $K_S$, the $t$ parties in $S$, iteratively decrypt the message using their private keys as follows. At step $i$, party $i$ updates:
\[
(C^{(i)}_1, C^{(i)}_2) = (C^{(0)}_1, C^{(i-1)}_2 - k_i C^{(0)}_1) 
                = (rB, X + rK_S - rK_i)
\]

Party $t$ then retrieves the message as $X = C_2^{(t)}$. Notice that at no point is the private key of any party shared with any other party. Furthermore, if all parties do not participate, then the message cannot be decrypted.

\descr{Re-encryption under another key.}
The $t$ servers can also collectively re-encrypt a message, encrypted under the collective key $K_\mathsf{S}$, to be under any public key $U$~\citep{Froelicher2017}. Given the encryption of the plaintext $X$ (mapped on the curve), as  $(C_1, C_2) = (rB, X + rK_S)$, the re-encryption is as follows. Define $(C^{(0)}_1, C^{(0)}_2) = (0, C_2)$. Each party $S_i \in S$ then samples $v_i$ uniformly at random from $[1, m-1]$ (nonce), and updates:
\[ 
(C^{(i)}_1, C^{(i)}_2) = (C^{(i-1)}_1 + v_iB, C^{(i-1)}_2 -rK_i +v_iU) 
\]

The $t$th server $S_t$ after performing the above computation, obtains the final encryption under $U$ as:
\[ 
(C^{(t)}_1, C^{(t)}_2) = (vB, X + vU) 
\]
where $v = v_1 + \cdots + v_t$. 
Notice that the re-encryption is done without any party being able to decrypt the content of the message.

\descr{Notation.} To avoid excessive notation, we shall write the encryption of a message $x \in [0, q-1]$ under public key $K$ as $\llbracket x \rrbracket_K$, where it is understood that fresh randomness is applied each time. Since the collective public key $K_\mathcal{S}$ shall be used for bulk of the encryption, we shall denote the encryption of $x$ under $K_\mathcal{S}$, i.e., $\llbracket x \rrbracket_{K_\mathcal{S}}$, simply as $\llbracket x \rrbracket$. 

\subsection{Assumptions and Limitations}
\label{subsec:assumptions}
We now discuss assumptions that we make in the desgin of DataRing and provide justification to each assumption.  

\descr{Unique Records in the Datasets.} We assume that the data domain $\mathcal{D}$ is public knowledge. For simplicity of analysis, we assume that each dataset $D$ is a collection of unique rows (regarding all attributes), i.e., there are no duplicate rows in the datasets. Thus if $N$ is the size of the dataset, then exactly $N$ labels have $\mathsf{val}() = 1$, and the remaining $|\mathcal{D}| - N$ labels have $\mathsf{val}() = 0$.

\noindent\textit{Justification:} For high-dimensional datasets, an overwhelming majority of the rows tend to be unique~\citep[\S 4]{Vincent2017}, and hence our analysis serves as a good approximation. We shall show in Section \ref{sec:exp_setup}, that the dataset~\citep{LendingClub} used in our system evaluation has 97\% of its of 2.26 million rows as unique. With this assumption, a dataset $D$ can be thought of a binary vector from the set $\{0, 1\}^{|\mathcal{D}|}$.

\descr{Background Knowledge of the Server.} 
We assume that the server $\mathcal{S}$ knows $L$ records of the participant $P$'s dataset, $D$. On the other hand, $P$ does not know which $L$ records are known by  $\mathcal{S}$. In other words, from $P$'s point-of-view, $\mathcal{S}$ knows a random sample of $L$ records from the dataset $D$. We call this the background knowledge $\mathcal{L}$ of $\mathcal{S}$ about $P$'s dataset. We also assume that sampling of $L$ records from dataset $D$ is done before the initialization of the data sharing, and is \textit{legally} mandated.

\noindent\textit{Justification:} Technically, this can be done by using a protocol similar to the partial view collection (cf. Sec \ref{subsec:pvc}), but with the assumption that participants do not deviate from their true datasets and servers can then decrypt the sampled subset $\mathcal{L}$. We refer the reader to the discussion in Section \ref{sec:intro} about the real-world considerations of the background knowledge of the server, as well as to the comparison with assumptions made in prior related work in Section \ref{sec:related_work}.   

\descr{On the Limits of Detecting Cheaters.} The ideal goal is to always detect any deviation from the true dataset. However, this is extremely challenging from a technical point of view. Instead, our guarantees of cheating detection are probabilistic. For instance, in order to ensure that the probability of avoiding detection is more than 0.95, the participant, with a dataset of size 500,000, needs to use 97.15\% of its true records in the partial view collection phase. Likewise, the participant needs to use its true dataset to answer \emph{all} queries to ensure that it avoids detection with a probability of more than 0.80. We believe that high profile organisations, e.g., banks, would prefer such cautious rates of avoiding detection (e.g., 0.95), due to the disproportionate nature of potential legal and reputational ramifications.  

\section{The DataRing}
\label{sec:system_model}

\subsection{System Overview}\label{subsec:system_overview}
Before giving a detailed description of DataRing, we briefly explain the system. The system consists of three main sequential phases as illustrated in Fig. \ref{fig:system_diagram}. 

\captionsetup[figure]{justification=centering}
\begin{figure}[t]
\centering
\includegraphics[width=0.8\linewidth]{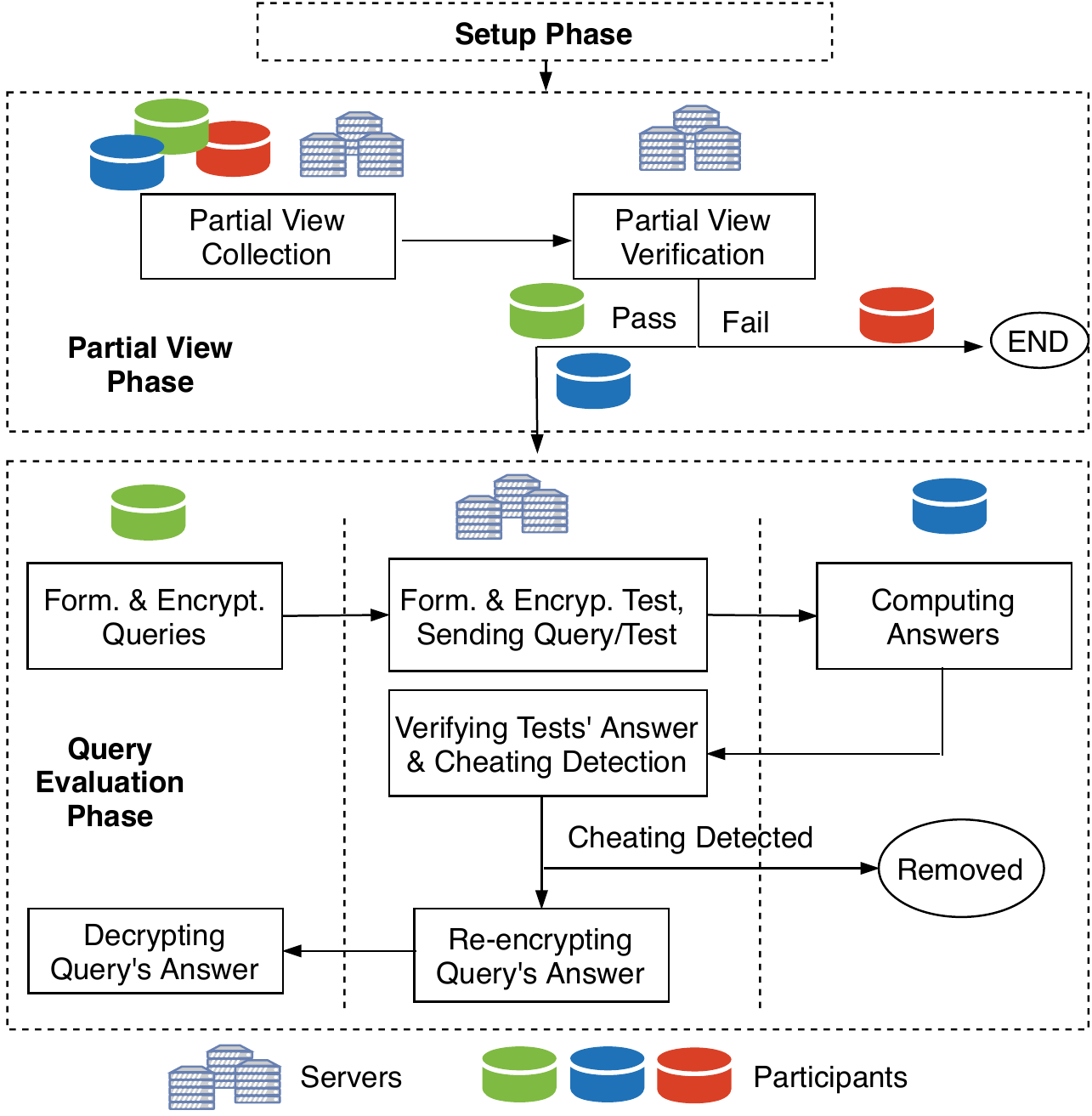}
\caption{DataRing's working flow. \label{fig:system_diagram}}
\end{figure}

\descr{Set-up Phase.} In this phase, all system parameters are initialized. Each participant publishes information such as its privacy budget and metadata. 

\descr{Partial View Phase.} 
In this phase, the server $\mathcal{S}$ collects a random subset of the participant $P$'s dataset (partial view) in such a way that $P$ does not know which records in its dataset are part of the partial view. 
Once the partial view has been obtained, the server $\mathcal{S}$ verifies if it is valid, i.e., if it is likely to have come from $P$'s true dataset $D$, based on the background knowledge that $\mathcal{S}$ knows about $D$.  
The participant can only proceed to next phase if its partial view is verified. The verified partial view is used as a reference to monitor the consistency of the participant's responses to queries. 

\descr{Query Evaluation Phase.}
In this phase, the server $\mathcal{S}$ ensures that each participant continues to use a dataset that is consistent with its partial view while answering other participants' queries. 
Specifically, $\mathcal{S}$ injects a number of \textit{hidden} test queries into the set of real queries that each participant must answer. Upon receiving participant's answers, $\mathcal{S}$ decrypts answers to the test queries to verify if they are consistent with $\mathcal{S}$'s background knowledge and the verified partial view. 
If answers to test queries are not consistent, the participant is detected as cheating. We call this process \textit{cheating detection}. The cheating participant is removed from the system and its answers are discarded. Answers to real queries from honest participants are re-encrypted by $\mathcal{S}$ using the honest querying participant's public key before being sent to it. Note that $\mathcal{S}$ never decrypts the answers to real queries. 

\subsection{Rationale Behind Construction}
There are alternative ways in which DataRing could be set up, which involve more participation from the servers. For instance, one alternative is to let the participant share its entire dataset with the server $\mathcal{S}$ albeit in its encrypted form (under the collective public key). $\mathcal{S}$ can then ensure the correctness of the entire dataset using its background knowledge. Moreover, $\mathcal{S}$ can compute answers to queries itself by adding appropriate differentially private noise via the Laplace mechanism (cf. Section~\ref{encrypted_question}), doing away with the additional query answer verification phase. However, there are two main issues with this approach. First, the participant now shares its entire encrypted dataset with the server $\mathcal{S}$, instead of an encrypted sample (partial view). Secondly, and more importantly, this modification requires the servers to generate Laplace noise to be added to the queries, thus revealing the exact noise added to the queries to the server. This partially violates the individual privacy and (query) confidentiality requirements outlined in Section~\ref{sub:desired}. In contrast, in our construction, the noise is added by the participant holding the dataset, and thus the servers never learn the approximate query answers or the noise added, apart from answers to the test queries, which only reveal what the servers already know, i.e., the background knowledge $\mathcal{L}$, the parial view size and the dataset size. 

\subsection{The Setup Phase}\label{setup}

The server $\mathcal{S}$ generates a collective public key $K_\mathcal{S}$ as discussed in Section~\ref{subsec:encryption}. 
The associated secret key is never reconstructed even when (collectively) decrypting any message encrypted under $K_\mathcal{S}$. 
Each participant $P_i$ generates its own public-private key pair and shares its public key with the server $\mathcal{S}$.
$P_i$ also publishes its privacy budget for every other participant and for the server $\mathcal{S}$ ($\epsilon_{i,j}$ and $\epsilon_{i,\mathcal{S}}$), its metadata (including dataset size $N_i$ and the domain $\mathcal{D}_i$), and the number of queries it wishes to evaluate on another participant $P_j$'s dataset ($|\mathcal{Q}_{i,j}|$). 

\section{Partial View Phase}
\label{sec:partial_view}

\subsection{Partial View Collection}
\label{subsec:pvc}
\subsubsection{Partial View Collection Protocol}
In the following, we describe how a partial view of size $V$ of a participant $P$'s dataset $D$ containing $N$ unique records is collected. 

\begin{enumerate}
\setlength{\itemsep}{1pt}
\item For each label $i$, participant $P$ creates a flag $f_i$. The flag is set to 1 if the corresponding label is \emph{in} the dataset ($\mathsf{val}(i)=1$)), otherwise it is set to 0.
The participant samples a random permutation $\sigma$ from $[|\mathcal{D}|]$ to itself, and creates the set $\mathsf{lbs} = \{(\sigma(i), f_i)\}_{i \in \mathcal{D}}$. 
The participant sends $\mathsf{lbs}$ to server $S_1$ and sends $\sigma^{-1}$ to server $S_2$. 

\item The server $S_1$ generates a random  $N$-element binary vector of Hamming weight $V$. This results in the $N$-element vector $\mathbf{u}$, given by $\mathbf{u} = \begin{pmatrix} u_1 & u_2 & \cdots & u_N 
\end{pmatrix}$
where exactly $V$ of the $u_j$'s are 1 and the rest are 0. Server $S_1$ then encrypts this using the key $K_\mathcal{S}$ as
\[
\llbracket \mathbf{u} \rrbracket = \begin{pmatrix} \llbracket u_1 \rrbracket & \llbracket u_2 \rrbracket & \cdots & \llbracket u_N \rrbracket
\end{pmatrix}
\]

\item For each pair $(j, f)$ in $\mathsf{lbs}$, the server $S_1$ does as follows. If $f = 1$, it pops an element $\llbracket u \rrbracket$ of $\llbracket \mathbf{u} \rrbracket$ and replaces $(j, f)$ with $(j, \llbracket u \rrbracket)$. Otherwise, if $f = 0$, it replaces $(j, f)$ with $(j, \llbracket 0 \rrbracket)$, where $\llbracket 0 \rrbracket$ is a fresh encryption of $0$ (under $K_\mathcal{S}$). The server $S_1$ then sends the modified $\mathsf{lbs}$, denoted $\llbracket \mathsf{lbs} \rrbracket$, to $S_2$.

\item The server $S_2$ updates each pair $(j, \llbracket f \rrbracket)$ of $\llbracket \mathsf{lbs} \rrbracket$ with $(j, \llbracket f \rrbracket + \llbracket 0 \rrbracket )$ effectively re-randomizing all encryptions. Server $S_2$ then further updates $\llbracket \mathsf{lbs} \rrbracket$ by applying the inverse permutation $\sigma^{-1}$ on each pair $(j, \llbracket f \rrbracket)$ as $(\sigma^{-1}(j), \llbracket f \rrbracket)$.
Finally $S_2$ shares the partial view (PV) with server $S_1$ defined as:
\[
\text{PV} = \begin{pmatrix}
\llbracket u_1 \rrbracket & \llbracket u_2 \rrbracket & \cdots & \llbracket u_{|\mathcal{D}|} \rrbracket
\end{pmatrix}
\]
where $u_i$ is the encrypted flag of the $i$th label in $ \mathsf{lbs} $. Note that we now have the original enumeration of the labels. 
\end{enumerate}

These steps are pictorially represented in Fig.~\ref{fig:pv} with simplified notation. Once the servers have obtained PV, they delete $\mathsf{lbs} = \{(\sigma(i), f_i)\}_{i \in \mathcal{D}}$ and $\sigma^{-1}$ which are never reused.

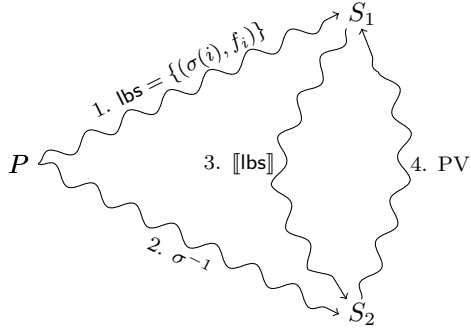
\begin{figure}[h]
\centering
\begin{tikzpicture}
\draw [->,decorate,decoration={snake,amplitude=1mm,segment length=6mm,post length=1mm}] (0,2) -- (4,4) node[pos=0, left] {$P$} node [pos=.5, sloped, right, above] {\footnotesize{1. $\mathsf{lbs} = \{(\sigma(i), f_i)\}$}} node[right] {$S_1$};
\draw [->,decorate,decoration={snake,amplitude=1mm,segment length=6mm,post length=1mm}] (0,2) -- (4,0) node[pos=0, left] {$P$} node [pos=.5, sloped, right, below] {\footnotesize{2. $\sigma^{-1}$}} node[right] {$S_2$};
\draw [->,decorate,decoration={snake,amplitude=1mm,segment length=6mm,post length=1mm}] (4.1,3.8) .. controls (3,2) .. (4.1,0.2) node [pos=.5, left] {\footnotesize{3. $\llbracket \mathsf{lbs} \rrbracket$}};
\draw [->,decorate,decoration={snake,amplitude=1mm,segment length=6mm,post length=1mm}] (4.3,0.2) .. controls (5,2) .. (4.3,3.8) node [pos=.5, right] {\footnotesize{4. PV}};
\end{tikzpicture}
\caption{Steps involved in Partial View Collection.}
\label{fig:pv}
\end{figure}

\begin{rem}\label{rem:pir_pv}
The idea behind our partial view collection scheme is reminiscent of Private Information Retrieval (PIR) technique \citep{pir_sergey}. 
However, PIR techniques are not applicable to our setting. In our setting, the server $S$ does not know the identities of actual elements of the participant's dataset $D$ which is essential for PIR to work. $S$ might use the public knowledge of all possible elements from the domain $\mathcal{D}$ and apply PIR to sample the partial view. Nonetheless, it could end up with drawing all elements that are not in $D$. Moreover, existing PIR schemes either allow the user to retrieve a single element at once \citep{SPIR} or require replicated databases \citep{PIR_1995} and are computationally expensive as a database must process all of its entries~\citep{xpir}.
\end{rem}

\subsubsection{Security Analysis}
We require that a cheating participant $P$ does not learn which of the $V$ records in its announced dataset $D'$ (potentially different from $D$) are in the partial view. Hence, the partial view, from the point-of-view of $P$, is a random sample of size $V$ of the $N$ records in $D'$. We also require that the dataset $D'$ is well-formed, i.e., has exactly $N$ unique records. 
On the other hand, we require that the servers in $\mathcal{S}$ do not learn which of the $N$ records in the domain are part of the dataset, and the random $V$-element sample itself.
The only requirement is that the servers should be able to jointly decrypt the ciphertexts of the labels from $\mathcal{L}$ in PV. 

In Appendix \ref{app:pvc}, we provide detailed proofs of the security of the partial view collection phase using the simulation paradigm~\citep{sim-tut}, where we define an ideal functionality for collecting the partial view. 

\begin{thm}
If EC-ELGamal is semantically secure under the decisional Diffie Hellman (DDH) assumption, our protocol in Section~\ref{subsec:pvc} securely collects the partial view.
\end{thm}

The above result is under the assumption that while the participant $P$ may deviate from its true dataset $D$, it does produce a valid permutation-inverse permutation pair $(\sigma, \sigma^{-1})$. The justification for this assumption is due to Theorem~\ref{the:fakeD}. In short, the adversary gains no advantage in using an invalid inverse permutation.

\begin{thm}
\label{the:fakeD}
Let $\mathcal{B}$ be an adversary which given the domain $\mathcal{D}$, the dataset $D$, and parameters $N$, $V$ and $L$, outputs $\{ (i,f_i)\}_{i \in \mathcal{D}}$, $f_i \in \{0, 1\}$ with $\sum_i f_i = N$, and permutation $\pi$, resulting in the dataset $D'$. Then there exists an adversary $\mathcal{A}$, given the same inputs, which chooses a random permutation $\sigma'$, and outputs $\{ (\sigma'(j),f'_j)\}_{j \in \mathcal{D}}$ and the inverse permutation $\sigma'^{-1}$, resulting in the same dataset. 
\end{thm}

\begin{proof}
For each $i$, set $j = \pi(i)$, and set $f'_j = f_i$. Sample a random permutation $\sigma'$. Output $\{ (\sigma'(j),f'_j)\}_{j \in \mathcal{D}}$, and $\sigma'^{-1}$. Clearly, the two result in the same dataset.  
\end{proof}

\subsection{Partial View Verification}
\subsubsection{Overview}\label{subsubsec:overview}

The partial view is used as a reference for monitoring the participant's answers to queries in the Query Evaluation Phase (cf. Section~\ref{sec:query_computation}). Thus, the server needs to verify if the partial view is indeed likely to have come from the true dataset $D$. This is mainly accomplished by the server checking if at least a threshold number of records in its background knowledge $\mathcal{L}$ are present in the partial view. 
 
Formally, the partial view verification phase is as follows:

\begin{enumerate}
    \item For each label $l$ in their background knowledge $\mathcal{L}$, the servers in $\mathcal{S}$ jointly decrypt the ciphertext $\llbracket u_l \rrbracket$ in PV.  
    \item If all the decrypted values are 0 or 1, and if at least $r_0 \in [1,L]$ of the $L$ records in $\mathcal{L}$ are found in PV,  $\mathcal{S}$ allows $P$ to proceed to the next phase. Otherwise, $P$ is rejected. 
\end{enumerate}

\subsubsection{Security Analysis} 
Since the partial view is a random sample of a participant's dataset, not all $L$ records in $\mathcal{S}$'s background knowledge ($\mathcal{L}$) are expected to be in PV. 
Let $R$ be the random variable denoting the number of records in $\mathcal{L}$, found in PV. From $P$'s view point, $L$ records in $\mathcal{L}$ are randomly sampled from the true dataset of $P$. As a result, $R$ can be viewed as a discrete random variable following a hypergeometric distribution with parameters $N, V$ and $L$ \citep{hypergeo_book,rice2006mathematical}. 

\descr{Balancing True and False Positives.} Even if the participant's dataset is the true dataset, there is a non-zero probability that none of the $L$ records are present in PV. Such a partial view would be falsely rejected. Thus, the server $\mathcal{S}$ checks if at least a threshold $r_0$ number of records from $\mathcal{L}$ are present in PV to balance between true and false positives. 

With this, the probability that $P$'s partial view passes verification is given by
\begin{equation}\label{eq:R_r_0}
    \mathrm{Pr}(R\geq r_0) = \sum_{r=r_0}^L \frac{\binom{V}{r}\binom{N-V}{L-r}}{\binom{N}{L}}
\end{equation}

\descr{Choosing $r_0$}. We fix a tolerated false positive rate $\eta$ (i.e., probability of falsely rejecting a partial view from honest participant). Then, $\mathcal{S}$ chooses $r_0$ such that the partial view from an honest participant is likely to include at least $r_0$ records from $\mathcal{L}$ with probability of at least $1-\eta$. 
From Eq.~(\ref{eq:R_r_0}), this means
\begin{equation}\label{eq:PR_V}
    \mathrm{Pr}(R \geq r_0) = \sum_{r=r_0}^L \frac{\binom{V}{r}\binom{N-V}{L-r}}{\binom{N}{L}} \geq 1- \eta
\end{equation}

\begin{rem}
When $\eta = 0$, the solution for $\mathrm{Pr}(R \geq r_0) = 1$ is $r_0 = 0$. This means an honest participant will pass verification with probability 1. On the other hand, this also means that a dishonest participant will also pass the verification. We, therefore, choose $\eta \in (0, 1)$, which means that $r_0 \ge 1$.   
\end{rem}

From this, $r_0$ is calculated as 

\begin{equation}\label{eq:fp_r0}
    r_0 = \max \left\{r\in [1, L] : \sum_{r=r_0}^L \frac{\binom{V}{r}\binom{N-V}{L-r}}{\binom{N}{L}} \geq 1-\eta \right\}
\end{equation}

\descr{Lower bound of the background knowledge's size.} For given dataset size $N$, partial view size $V$ and a tolerated false positive rate $\eta$, to ensure the threshold $r_0 \ge 1$ there is a lower bound for the size of the background knowledge. Note that from the condition for a partial view passing the verification $Pr(R\ge r_0) \ge 1 - \eta$, we have $r_0 \ge 1 \iff Pr(R = 0) < \eta$. Thus, to determine the absolute minimum value of $L$ we end up with: $$L_{\text{min}} = \min \left\{L \in [1,N] : \prod_{k=0}^{V=\rho N}{\frac{(N-k-L)}{(N-k)}} < \eta \right\}$$

\begin{rem}\label{rem:size_of_backgroundKnow}
Crucially, the minimum number of records in the background knowledge $L_{\text{min}}$ does not increase with the dataset size, as long as the ratio between the partial view size and the dataset size ($\rho = V/N$) remains fixed. Thus, the servers do not need to sample a larger $L_{\text{min}}$ for larger datasets. For instance, if $\eta = 0.05$ and $\rho = 0.01$, the lower bound remains $L_{\text{min}} = 299$ for dataset sizes $N = [500K, 1M, 1.5M, 2M]$.
\end{rem}

\section{Query Evaluation Phase}
\label{sec:query_computation}

\subsection{Overview}\label{subsec:overview_query_phase}
In this phase, all participants passing partial view verification proceed to analyze each other's datasets. The analysis consists of a set of queries and their answers. The server $\mathcal{S}$ actively participates in the data analysis process to ensure that the analysis is consistent with the (already verified) partial view in a privacy-preserving manner. $\mathcal{S}$ ensures this by conducting a ``cheating detection'' process through hidden tests during this phase. 

Assume that participant $P_i$ wishes to analyze participant $P_j$'s dataset $D_j$, facilitated by $\mathcal{S}$. The query evaluation phase consists of the following steps.

\begin{enumerate}
    \item Participant $P_i$ sends $m_\mathsf{q}$ encrypted queries to the server $\mathcal{S}$ to be evaluated on $P_j$'s dataset, denoted $\mathcal{Q}_{i,j} = \{Q_1,\ldots, Q_{m_\mathsf{q}}\}$.
    
    \item To monitor whether $P_j$ computes queries' answers using a dataset consistent with its partial view, the server $\mathcal{S}$ conducts \emph{hidden tests}. More specifically, $\mathcal{S}$ generates and encrypts $m_\mathsf{t} \leq m_\mathsf{q}$ test queries for $P_j$, denoted $\mathcal{T}_j = \{T_1,\ldots, T_{m_\mathsf{t}}\}$, based on its background knowledge about $D_j$ and $P_j$'s submitted partial view. $\mathcal{S}$ knows the true answer of each test. However, the answers returned by the participant are noisy, and hence exact answers to test queries are not obtained. 
    Thus, we define ``expected answer'' of a test query as any value that does not differ from the true answer by more than a maximum noise amount that can be added to a query answer (cf. Section \ref{subsec:test_functions}). 
     
     \item $\mathcal{S}$ bundles the real queries in $\mathcal{Q}_{i,j}$ and test queries in $\mathcal{T}_j$ together (cf. Section \ref{encrypted_question}), resulting in a total of $m = m_\mathsf{q} + m_\mathsf{t}$ queries. $\mathcal{S}$ then randomly pops a query from these $m$ queries and sends it to $P_j$ at each iteration. Thus, the order of test queries and real queries sent to the participant is random.
     
     \item Upon receiving an encrypted query, $P_j$ evaluates it on its dataset, adds Laplace noise of scale $m_q/\epsilon$ to the answer and then sends the final answer back to $\mathcal{S}$. Note that $P_j$ has separate privacy budgets for both the server ($\epsilon_{j,S}$) and the participant $P_i$ ($\epsilon_{j,i}$) with $\epsilon_{j,S} = \epsilon_{j,i} = \epsilon$, and therefore, the privacy leakage of $P_j$'s dataset is within the overall budget specified.
    
     \item Answers from $P_j$ are held at the server $\mathcal{S}$. When all query answers are obtained, $\mathcal{S}$ decrypts answers of test queries to verify if they are the expected answers. If $P_j$ fails to provide expected answer to any test query, it is detected as cheating. 
     Otherwise, $P_j$'s answers of real queries are re-encrypted under the participant $P_i$'s public key. 
     
\end{enumerate}

\begin{rem}\label{rem:query_compute}
In the query evaluation phase, encrypted queries are sent to the participant who is able to respond to the query without learning its content. This may appear similar to the PIR setting~\citep{pir_sergey, xpir}. However, PIR scheme is not applicable in our setting as it either requires the dataset to be sent to the querying participant or requires the querying participant to know the exact indexes of the data records that they wishes to retrieve. In our setting, a participant's dataset is never shared with the servers and other participants.
\end{rem}

\subsection{Test Queries}\label{subsec:test_functions}

The analysis in this section applies to all participants. We, therefore, remove the subscripts to identify individual participants. Let $T$ be a test query, and let $\epsilon$ denote the privacy budget of a participant for the server.
As before, $m_\mathsf{q}$ denotes the number of real queries. The number of test queries is upper bounded by the number of real queries, $m_{\mathsf{t}} \le m_{\mathsf{q}}$. This is because, the participant should not be able to distinguish between the real and test queries, and therefore, should not spend more than the allocated privacy budgets on $m_{\mathsf{q}}$ and $m_{\mathsf{t}}$.
Let $\Delta T$ denote the sensitivity of the test query. Since all queries are count queries, we have $\Delta T = 1$. Let `ans' denote the noisy answer to the test query $T$ from the participant. Thus, the server defines $\text{noise}_{\max} = \ln(1/\delta)\frac{m_{\mathsf{q}} \Delta T}{\epsilon}$. This means that the (Laplace) noise added to a test query answer is less than $\pm \text{noise}_{\max}$ with probability $1 - \delta$~\citep[\S 3.3]{Dwork2014}. 
We have three broad categories of tests (test queries) to counter several cheating strategies. 

\descr{Tests based on the Servers' Background Knowledge ($\mathcal{L}$).}  This test aims at confirming the presence of $L$ records from $\mathcal{L}$ in $P$'s dataset. Specifically, the test query  counts the number of records in $P$'s dataset that match records in $\mathcal{L}$. The acceptance range of an answer to this test is $ \text{ans} \in [L \pm \text{noise}_{\max}]$. We call this \textit{Test $L$}. 

\descr{Tests based on Partial View.} This hidden test aims at verifying if $P$ computes answers to test queries using a dataset consistent with its partial view. The test query counts the number of records in $P$'s dataset matching records in $P$'s partial view. If the participant is honest, this should exactly equal to $V$ modulo some noise. Thus, the acceptance range of an answer to this test is $\text{ans} \in [V \pm \text{noise}_{\max}]$. We call this test query \textit{Test $V$}.

\descr{Tests based on the Dataset Size.} 
This test aims at detecting a cheating participant that adds arbitrary records to the true dataset. The test query counts the total records in the dataset. As $\mathcal{S}$ knows the exact size of $P$'s dataset, the acceptance range of an answer to this test is $\text{ans} \in [N \pm \text{noise}_{\max}]$. We denote this test query as \textit{Test $N$}.

\descr{Justification for Test Queries.} All three test queries are necessary.
\textit{Test $L$} alone is not enough to counter all cheating strategies because the participant can simply return all noisy answers (even to real queries) close to $L$.
Similarly, using only \textit{Test $V$} can detect the cheater answering noisy answers close to $L$ but still cannot detect the cheater who evaluates queries on a larger dataset (adding additional fake rows beyond the actual number of rows $N$). Using \textit{Test $N$} alone cannot  detect a cheater who always uses a fake dataset which has same size as the true dataset. Thus, we use the combination of all test types to detect cheating participants.

\subsection{Query Formulation \& Answer Computation}\label{encrypted_question}
This section describes the formulation of a query (real query and test query), its encrypted version and how the query answer is computed.

\descr{Encrypted query formulation.} Let $Q \in \mathcal{Q} \cup \mathcal{T}$ denote a test or a real query. The query $Q$ in vector form is $Q = \begin{pmatrix} q_1 & q_2 & \cdots & q_{|\mathcal{D}|} 
\end{pmatrix}$,
where $q_i = 1$ if label $i$ is to be counted, and $0$ otherwise. Recall that mapping a label (data point) to be evaluated in a query is based on the publicly known ordering (cf. \ref{subsec:dataset}). The queries are encrypted as
\[
\llbracket Q \rrbracket = \begin{pmatrix} \llbracket q_1 \rrbracket & \llbracket q_2 \rrbracket & \cdots & \llbracket q_{|\mathcal{D}|} \rrbracket \end{pmatrix}
\]
Note that the real queries are encrypted by the participant (who wishes to analyse another participant's data) and then sent to the server.
Thus, none of the servers in $\mathcal{S}$ learn the content of the real queries. The server then iteratively pops a random query from $\mathcal{Q} \cup \mathcal{T}$ and sends it to the targeting participant. 

With this formulation, a test query $T = \begin{pmatrix} q_1 & q_2 & \cdots & q_{|\mathcal{D}|}\end{pmatrix}$ is formed as follows. For the \textit{Test $L$} query, $q_i = 1$ if and only if $i$ is one of the $L$ labels in the server's background knowledge. For the \textit{Test $N$} query, $q_i = 1$ for all $i$. The server $\mathcal{S}$ then encrypts $q_i$ using $K_\mathcal{S}$ to obtain corresponding encrypted test $\llbracket T \rrbracket$.
If $T$ is the \textit{Test $V$} query, the server $\mathcal{S}$ sets $\llbracket q_i \rrbracket = \llbracket u_i \rrbracket + \llbracket 0 \rrbracket$, for each $i$, where $\llbracket u_i \rrbracket$ is the $i$th element of the (encrypted) partial view, and $\llbracket 0 \rrbracket$ is the fresh encryption of $0$. Then, $\mathcal{S}$ sets $\llbracket T \rrbracket = \begin{pmatrix} \llbracket q_1 \rrbracket & \cdots & \llbracket q_{|\mathcal{D}|}  \rrbracket 
\end{pmatrix}$ as the encrypted \textit{Test $V$} query, which is the re-encryption of the partial view. 
Note that $\mathcal{S}$ can reuse a test query by re-randomizing all the encryptions in the encrypted test $\llbracket T \rrbracket$ by adding $\llbracket 0 \rrbracket$ to them. 

\descr{Answer computation.} Once participant $P$ receives an encrypted query $\llbracket Q \rrbracket$ from $\mathcal{S}$, it computes query answer on its dataset $D$ as follows
\begin{enumerate}
    
    \item It initializes $\llbracket \text{ans} \rrbracket \leftarrow 0$. For each $i \in \mathcal{D}$ such that $\mathsf{val}(i) = 1$, it updates $\llbracket \text{ans} \rrbracket \leftarrow \llbracket \text{ans} \rrbracket + \llbracket q_i \rrbracket$.
    \item It draws $\text{noise} \leftarrow \text{Lap}\left(\frac{m_\mathsf{q}}{\epsilon}\right)$, encrypts it as $\llbracket \text{noise} \rrbracket$, and updates $\llbracket \text{ans} \rrbracket \leftarrow \llbracket \text{ans} \rrbracket + \llbracket \text{noise} \rrbracket$.
    \item The participant sends $\llbracket \text{ans} \rrbracket$ to the server $\mathcal{S}$.
\end{enumerate}
 
Note that although the query includes $|\mathcal{D}|$ elements, the participant only needs to process $N$ of them corresponding to its dataset $D$, since $\langle Q, D \rangle = \sum_{i : \mathsf{val}(i) = 1} q_i$.

\subsection{Query Answer Release} \label{subsec:answer_release}
The server $\mathcal{S}$ conducts cheating detection with all participants in the system.
After the query evaluation of any two participants on each other's dataset, $\mathcal{S}$ releases answers to each participant based on the output of the cheating detection process. If no cheating is detected from both participants, $\mathcal{S}$ re-encrypts all answers to under inquiring participants' public keys. $\mathcal{S}$ then sends re-encrypted answers to the inquiring participants which can decrypt the answers using their secret keys. If any participant is detected as cheating, it will be removed from the system. Then, all query answers of both participants are discarded.  Thus, answers from an honest participant are only sent to honest queriers.

\subsection{Cheating Detection Analysis}
\label{subsec:cheating_detection}

Given the probabilistic nature of EC-ElGamal encryption, all encrypted queries are indistinguishable. More formally, we prove the following theorem in Appendix~\ref{app:tvsr}.
\begin{thm}
If EC-ElGamal is semantically secure, then the encrypted queries from $\mathcal{Q}$ are indistinguishable from encrypted queries from $\mathcal{T}$.
\end{thm}

Now, given that the queries are indistinguishable, a cheating participants strategy is to randomly return incorrect answers to one or more queries and ``hope'' that they are not test queries. We make this more precise in the following. At iteration $j$, let $Q_j \in \mathcal{Q} \cup \mathcal{T}$ denote the query sent to the participant. Let $p_\mathsf{t}$ be the probability that $Q_j$ is a test query. Let $p_\mathsf{c}$ denote the probability that a cheating participant decides to cheat at iteration $j$. Finally, let $p_\mathsf{d}$ be the probability that the cheating participant is caught at this iteration. Notice that $p_\mathsf{d}$ is defined over all test queries and possible cheating strategies deployed by the cheating participant. 

We are interested in the probability that the server $\mathcal{S}$ successfully detects a cheating participant $P_c$ after monitoring $m$ queries in $\mathcal{Q} \cup \mathcal{T}$, denoted by $\Pr(\text{success})$.
Let $A_j$ be the event that $P_c$ is not caught at iteration $j$. Then 
\[
\Pr (A_j) = 1- p_\mathsf{t} p_\mathsf{c} p_\mathsf{d}.
\]

From this, $\Pr(\text{success})$ is given as
\begin{equation} 
    \Pr (\text{success}) = 1 - \prod_{j=1}^{m} \Pr \left(A_j\right) = 1 - (1- p_\mathsf{t} p_\mathsf{c} p_\mathsf{d})^m
    \label{eq:prob_detect_lie}
\end{equation}

Given $m$, $p_t$ and $p_d$, the more number of times a participant cheats during the query evaluation (i.e., the higher $p_c$), the higher the probability that the server successfully detects cheating. In contrast, when the participant cheats a few times (i.e., very small $p_c$), it is more likely that the participant provides incorrect answers to real queries in which the server cannot verify. Consequently, the probability that the server successfully detects this cheating might be small. Our cheating detection aims at minimising the number of times that a participant can cheat during the query evaluation. 

We choose the experimental route in Section \ref{subsec:misbehaved_participant}, where we analyze the probability that the server successfully detects a cheating participant given the number of times the participant cheats using specific cheating strategies and the type of tests used by the server. 

\section{Evaluation Setup}
\label{sec:exp_setup}
We implemented DataRing in C/C++ using the C implementation of the additive homomorphic Elliptic Curve ElGamal cryptosystem~\citep{elgamal} as a base. 
We used ElGamal encryption on prime256v1 elliptic curve with 128-bit security as defined in OpenSSL \citep{openssl_prime256, rfc_prime256_secp256r1}.

We ran our experiments on Amazon EC2 \citep{amazon_ec2} using r5.4xlarge instance with 16 cores and 128GB of memory. To validate our theoretical analysis, we used a real-world dataset containing information regarding 2.26 million loans made on a peer-to-peer lending platform called Lending Club \citep{LendingClub}. 
We chose 10 attributes related to a borrower including loan amount, term, interest rate, etc. 
For categorical attributes, we encode their values using integer encoding.  
We observe that in the dataset of 10 chosen attributes and 2.26 million records, 97\% of records are unique. 
Thus, there are 2.2 million unique records from this dataset. 
We extract each participant's dataset $D$ from the 2.2 million unique records. 
For each participant's dataset $D$ thus extracted, we assume that the server $\mathcal{S}$ knows a random subset of $L$ records as the background knowledge.

It is noted that this work considers static datasets and leaves evolving datasets for future work.  

\begin{table}[ht!]
\centering
\caption{Default parameter values used for the evaluation} \label{table:table:default_parameter}
\begin{tabular}{ l c c } 
\hline
Parameter &  Value \\
\hline
Participant's dataset size $(N)$ &  500,000\\
Domain cap $(a)$ &   4\\
Domain size $(|\mathcal{D}|)$ &   2,000,000\\
Partial view size (V)  &  5,000\\
Servers' background knowledge size $(L)$ &  500\\
Privacy budget of a party for another party $(\epsilon)$ & 0.5\\
Privacy budget of a party for $\mathcal{S}$ $(\epsilon_{\mathcal{S}})$ & 0.5\\
\hline
\end{tabular}
\end{table}

\descr{Reducing the Domain Size.} Recall a dataset $D$ is represented as a histogram over a public domain $\mathcal{D}$. If the domain size is very large, which is likely to be the case with high dimensional datasets, it imparts an exponential penalty on time and communication complexity. For efficiency, we impose a cap on the size of the domain as a multiple $a > 1$ of the size of the dataset. We call the integer $a$, the \emph{domain cap}. Thus for our experimental evaluation, for a dataset of size $N$, the domain size, i.e., the number of possible records are assumed to be $aN$, a subset of all domain points from $\mathcal{D}$. This includes the $N$ records in the dataset. Identifying which points to include in the capped domain can be done algorithmically by sampling points which are close to actual points. The rest of the domain points are discarded. 

While this is not ideal from a privacy point of view (an adversary, e.g., the server, can know which domain points are \emph{not} part of the dataset), it is not a blatant compromise of privacy. First, in many real-world datasets, a large portion of domain points, so-called \emph{structural zeroes}, are never realized in practice (e.g., an interest rate disproportionate to the loan amount lent). Thus, these may well be discarded from the domain. Second, even with the restricted domain size of $aN$, there are $\binom{aN}{N}$ possible datasets, the dataset $D$ being one of them. Assuming $N$ to be large, the resulting set of candidate datasets is too large for the dataset $D$ to be identified. With these considerations in mind, we use the domain cap as a trade-off between privacy and efficiency.

The default parameter values for the evaluation of DataRing are summarized in Table \ref{table:table:default_parameter}. We consider the ratio between dataset size and partial view size $\rho=0.01$ for given $\eta = 0.05$. Thus, the minimum size of server's background knowledge is $L_{min}=299$. We use the server's background knowledge size $L=500$ in all evaluations which ensures the condition for a partial view passing the verification $r_0 > 1$.

\section{Security Evaluation}
\label{sec:evaluation}
In this section, we evaluate the robustness of DataRing in defending against cheating participants. 

\subsection{Effectiveness of Partial View Collection \& Verification}\label{subsec:robustness_pv}

According to the partial view collection, the participant is supposed to send the set $\mathsf{lbs} = \{ (\sigma(i),f_i)\}_{i \in \mathcal{D}}$ to $\mathcal{S}_1$ and the corresponding inverse permutation $\sigma^{-1}$ to $\mathcal{S}_2$ for sampling a partial view of $V$ records (each record is represented by a label $i$). 
According to Theorem \ref{the:fakeD}, we can simply concentrate on an adversary that chooses a dataset at the start (possibly fake), a random permutation $\sigma$ and its correct inverse $\sigma^{-1}$. What can be done with any other strategy, e.g., by submitting an inverse permutation different from $\sigma^{-1}$, can also be done with this strategy with the same advantage. 

Let $P_c$ denote a cheating participant with true dataset $D$, who creates a fake dataset $D'$. The fake dataset $D'$ is parameterised by $n$, the number of true records from $D$ kept in $D'$, where $0 \le n \le N$. The two extremes are when $n = 0$, a completely fake dataset, and when $n = N$, the true dataset. Participant $P_c$ can create a fake dataset by choosing $n$ true records, and $N - n$ records from the rest of the domain, i.e., which are not part of $D$. 
The aim of $P_c$ is to choose $D'$ in such a way that the partial view sampled from $D'$ includes as little true data as possible while still passing the verification. Thus, $P_c$ try to choose an $n$ as small as possible. In the below, we show how $P_c$ chooses such an $n$.

\descr{Finding the minimum value of $n$.} 
Among $V$ records sampled from $D'$ for the partial view, let $v$ denote the number of records in the partial view sampled from $n$ true records. Thus, if any records in $\mathcal{L}$ are found in the partial view, they must come from these $v$ records.  
Let discrete random variable $R_v$ denote the number of records in $\mathcal{L}$ found in $P_c$'s partial view. $R_v$ follows a hypergeometric distribution with parameters $N, v, L$. 

$P_c$ is interested in minimizing $n$. However, it must also ensure that at least $r_0$ records from $\mathcal{L}$ are present in its partial view with a target probability $\theta$. In other words, $P_c$ must ensure that $ \mathrm{Pr}(R_v \geq r_0) \geq \theta$. To ensure this, $P_c$ first determines $v_{{\min}}$ to satisfy these constraints, leading to the following definition:
\begin{equation}\label{eq:v_x_min}
    v_{{\min}} = \min \left\{v\in [0, V] : \sum_{k=r_0}^L \frac{\binom{v}{k}\binom{N-v}{L-k}}{\binom{N}{L}} \geq \theta \right\}
\end{equation}

Given parameters $N, V, L$ and its target $\theta$, based on Eq. (\ref{eq:v_x_min}) $P_c$ picks a value of $r_0$ from $[1, L]$ to determine $v_{{\min}}$. Note that for a given target $\theta$, $P_c$ could obtain a different $v_{{\min}}$ when it picks a different $r_0$. 
As $P_c$ is not aware of the exact threshold $r_0$ used by the servers, to meet its target $\theta$, the safer strategy for $P_c$ is to choose the maximum value of the obtained $v_{{\min}}$. We denote this value $v_{\text{opt}}$. 

Next, $P_c$ determines the minimum value of $n$ so that at least $v_{\text{opt}}$ records are sampled from these $n$ true records with probability $\theta$. This can be determined as: 

\begin{equation}\label{eq:n_x_min}
    n_{{\min}} = \min \left\{n \in [0, N] : \sum_{k=v_{\text{opt}}}^V \frac{\binom{n}{k}\binom{N-n}{V-k}}{\binom{N}{V}} \geq \theta \right\}
\end{equation}

Table \ref{table:n_opt_participant_pass_rate} shows the minimum number of true records of $D$ maintained in $D'$ ($n_{{\min}}$) against the participant $P_c$'s target probabilities $\theta$, given the tolerated false positive rate $\eta = 0.05$ and other parameters as in Table \ref{table:table:default_parameter}.
We can see that in order to pass the verification with high probability ($\theta \geq 0.9$), the participant must use a dataset $D'$ that contains substantial proportion of the true dataset $D$. For instance, with $\theta = 0.95$, the participant must maintain at least 485,786 (i.e., 97.15\%) true records of $D$ in $D'$. In other words, checking the validity of the collected partial view against the background knowledge of the server could effectively enforce the participant to submit a sample of its true dataset.

\begin{table}[t]
\centering
\captionof{table} {Number of true records $n_{\min}$ maintained against actual target probabilities ($\theta$)}

 \begin{tabular}{r|| c c c c} 
 
$\theta$ &0.91  &0.93  &0.95 &$\ge 0.96$ \\
 \hline
$n_{{\min}}$&410,730 &443,155  &485,786  &500,000 \\
\end{tabular}
\label{table:n_opt_participant_pass_rate}
\end{table}

\subsection{Robustness of Cheating Detection}
\label{subsec:misbehaved_participant}
Recall from Section~\ref{subsec:overview_query_phase} that a participant is given $m = m_{\mathsf{q}} + m_{\mathsf{t}}$ queries, $m_{\mathsf{q}}$ of which are real queries (from another participant) and $m_{\mathsf{t}}$ are test queries from the server. The server verifies answers to all test queries to detect a cheating participant after $m$ queries are evaluated. 

Let us consider ``cheating'' as the \textit{positive} class and ``honest'' as the \textit{negative} class in our cheating detection scheme. We define a true positive (TP) is an outcome where the server $\mathcal{S}$ correctly detects a cheating participant. A false positive (FP) is an outcome where $\mathcal{S}$ incorrectly detects an honest participant as cheating. A true negative (TN) is an outcome where $\mathcal{S}$ correctly detect a honest participant. A false negative (FN) is an outcome where $\mathcal{S}$ incorrectly detects a cheating participant as honest one. 

Our goal is to detect cheaters while minimising any adverse impact on honest participants, i.e., $\text{FP} = 0$. In order to do this, we set $\delta = 0.95$ when generating $ \text{noise}_{\max}$, i.e., 95\% of the noise values drawn from the $\text{Lap}(m_q/\epsilon)$ lie within $[ \pm\text{noise}_{\max}]$ (cf Sec. \ref{subsec:test_functions}). With this value of $\delta$, we achieved a 100\% TN (i.e., a 0\% FP). Since the honest participants are not affected, our accuracy metric is thus a direct measure of detecting cheaters:    
\[
\text{Accuracy} = \frac{\text{TPs}}{\text{TPs}+\text{FNs}} 
\]
To compute accuracy, we fix the number of incorrect answers from a cheating participant, and run the experiment 30 times. If cheating is detected, we increase TP by 1. Accuracy is thus the average over all 30 runs.

\subsubsection{Cheating Strategies}

The cheating participant's goal is to give incorrect answers to the $m_q$ real queries. We define an \textit{incorrect answer} as the answer to a query when the participant evaluates the query on a fake dataset. However, since these queries cannot be distinguished from the $m_t$ test queries, the participant cannot submit incorrect answers to all queries without being detected. Thus, a safer strategy for the cheating participant is to generate a fake dataset $D^*$ (which can be completely different from $D$ or a modified version of $D$). Once receiving a query, it randomly chooses to use the true dataset $D$ or the generated fake dataset $D^*$ to compute the query's answer. Hence, the number of incorrect answers varies from $[0:m]$.

We below envision two best strategies that a cheating participant could use to generate $D^*$. 
\begin{itemize}
    \item \textbf{Modifying original data records}: The goal of this strategy is to modify a fraction of the true dataset and hence give incorrect answers. This can be done by replacing a number of original data records in $D$ by new arbitrary records. Specifically, let $X$ denote the number of records in  $D$ to be replaced by arbitrary records, ($0 \leq X \leq N$). Hence, in the modified dataset $D^*$, the total number of data records remain at $N$, i.e., $\sum_{l} \mathsf{val}(l)=N$. We call the ratio of $X/N$ the \textit{modifying rate $\alpha$}, where $0 \leq \alpha \leq 1$.
    \item \textbf{Adding data records}: The goal of this strategy is to add arbitrary data records to the original dataset $D$ to obtain fake dataset $D^*$ and hence possibly scale up the query's answer. Let $\omega$ denote the \textit{adding rate}, ($\omega > 0$). The number of added data records is $\omega N$. Thus, the total number of data records in $D^*$ is $(\omega + 1)N$, i.e., $\sum_{l} \mathsf{val}(l)= (\omega+1) N$. 
\end{itemize}

\begin{figure}[ht]
\centering
\includegraphics[width=0.7\linewidth]{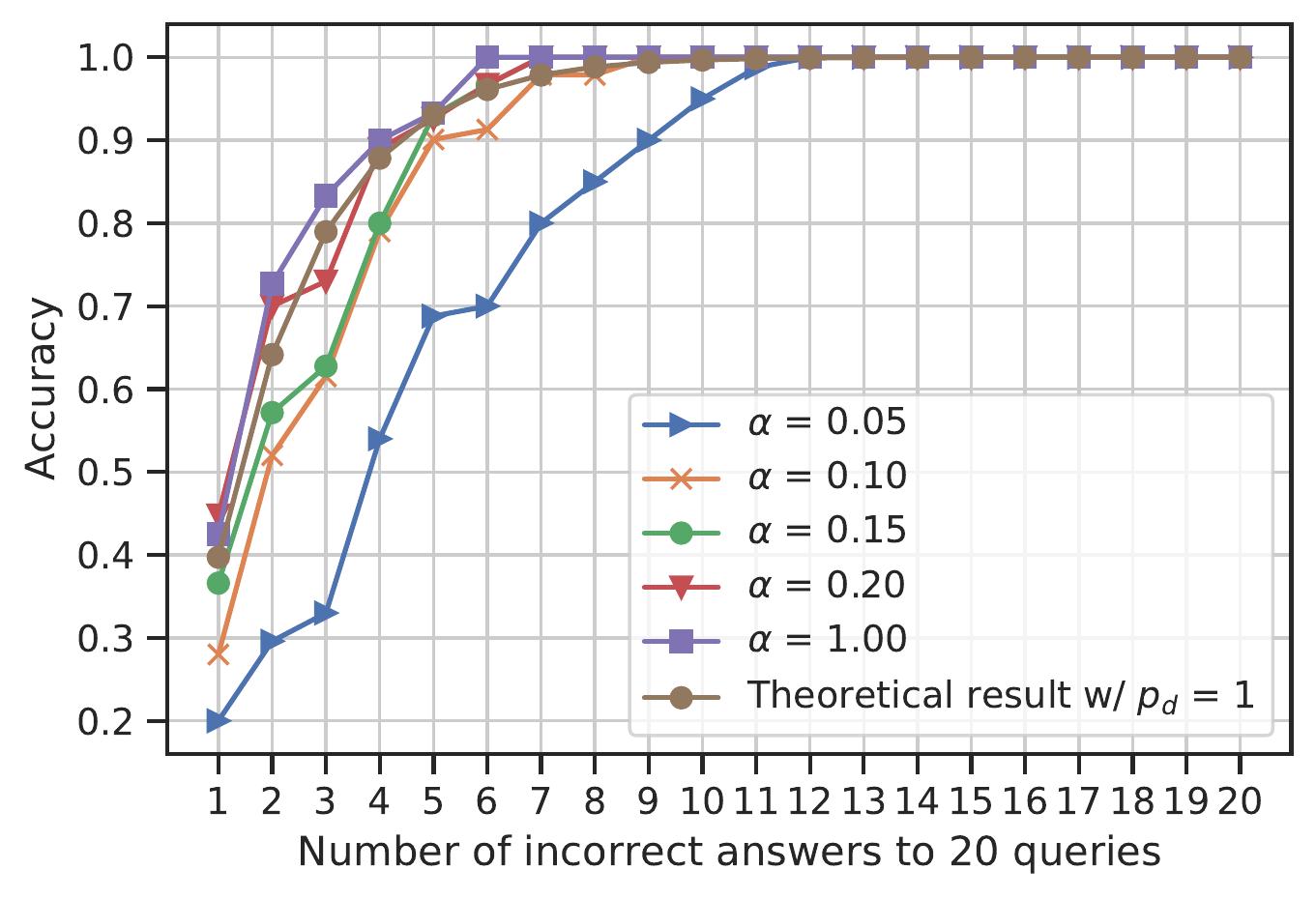}
\caption{Accuracy of the cheating detection against modifying data records strategy
\label{fig:test_anal_experiment_theoretial_replace}} 
\end{figure}

\subsubsection{Experimental Results} \label{subsubsec:experiment}

We consider the data sharing between two participants $P_1$ and $P_2$. 
As the server monitors the behavior of each participant separately, we only show the experimental results regarding $P_1$.

We assume that the server $\mathcal{S}$ receives $m_q = 10$ encrypted real queries from $P_2$ (to be evaluated on $P_1$'s dataset). It also generates and encrypts $m_t = m_q$ test queries to conduct cheating detection against $P_1$.  
Thus, there are $m=20$ queries to be evaluated on $P_1$'s dataset. Other parameters are set as in Table \ref{table:table:default_parameter}. 

\descr{Accuracy of the cheating detection.} 
In the following we shall show the experimental results when $P_1$ uses two main cheating strategies in the query evaluation in Fig. \ref{fig:test_anal_experiment_theoretial_replace} and Fig. \ref{fig:test_anal_experiment_theoretial_scaleup} respectively. Note that the total privacy budget is $\epsilon = \epsilon_{\mathcal{S}} = 0.5$.

We first show the experimental results when $P_1$ modifies original data records in $D$ to obtain $D^*$. It may be possible that the participant hides (removes) a small number of data records (potentially more valuable than others) by replacing them with arbitrary ones, i.e., $\alpha \approx 0$. However, our goal is to ensure that statistical count queries are answered close to their true answers, which will not be impacted by hiding a very small number of records. Also, note that the modifying rate $\alpha$ is similar to the poisoning rate in machine learning's poisoning attacks which has been considered to be less than 20\% in prior work \citep{poisoning_attack}. In an extreme case, $P_1$ evaluates queries on a completely fake dataset $D^*$, i.e., $\alpha = 1$. Thus, we evaluate the accuracy of our cheating detection scheme against  cheating participant that modifies the original dataset with various modifying rates $\alpha = [0.05, 0.1, 0.15, 0.2, 1]$.

Figure \ref{fig:test_anal_experiment_theoretial_replace} shows the accuracy of our cheating detection against the cheating participant that provides $x=[1:20]$ incorrect answers with different modifying rates $\alpha$. First, we can see given a fix number of incorrect answers, the cheating detection achieves higher accuracy when higher $\alpha$ values were used by $P_1$. This is because larger number of records are replaced by arbitrary data making the test query answers more likely to lie outside the acceptance range. Second, the accuracy is closer to the theoretical probability of the server successfully detects cheating with $p_\mathsf{d} = 1$ (computed based on Eq.~(\ref{eq:prob_detect_lie})). Third, with a small modifying rate, e.g., $\alpha = 0.05$, our cheating detection achieves an accuracy of lower than 50\% when less than 3 incorrect answers were provided. Nevertheless, when a small number records was modified, it has less impact on the statistical count queries. Next, for all investigated $\alpha$ values, the accuracy increases when more incorrect answers were provided and it reaches 100\% when the participant provides more than 11 incorrect answers since at least one incorrect answer is given to a test query and being caught. Finally, when $P_1$ uses $\alpha = 0.2$ and provides only one incorrect answer, our cheating detection can still detect it as cheating with an accuracy of 45\%. Thus, cheating participants cannot always escape detection even they deviate from the protocol only once. 

We now show the experimental results when $P_1$ adds arbitrary data records to the true dataset $D$ with adding rate $\omega =[0.5, 1.0]$ in Fig.  \ref{fig:test_anal_experiment_theoretial_scaleup}. Specifically, to obtain $D^*$, $P_1$ adds $\omega N$ extra data records to $D$. 
We can see that the accuracy is close to the theoretical results with $p_\mathsf{d} =0.5$, which considers the average success rate of a test is 0.5. Given the number of incorrect answers, the cheating detection achieves higher accuracy with higher adding rate $\omega$. Similarly, for a given $\omega$, the server achieves higher cheating detection accuracy when more incorrect answers were provided. It is also noted that although the cheating participant has higher chance to pass the cheating detection by adding data records as compared to modifying data records, it still cannot always get away undetected even it deviates from the protocol only once.

In summary, our results indicate that the cheating detection can catch a participant that provides incorrect answer to only \textit{one} query with a probability of 0.2. In the other words, the participant needs to use its true dataset to answer all queries to ensure that it avoids detection with a probability of more than 0.8.

\captionsetup[figure]{justification=centering}
\begin{figure}[t]
\centering
\includegraphics[width=0.7\linewidth]{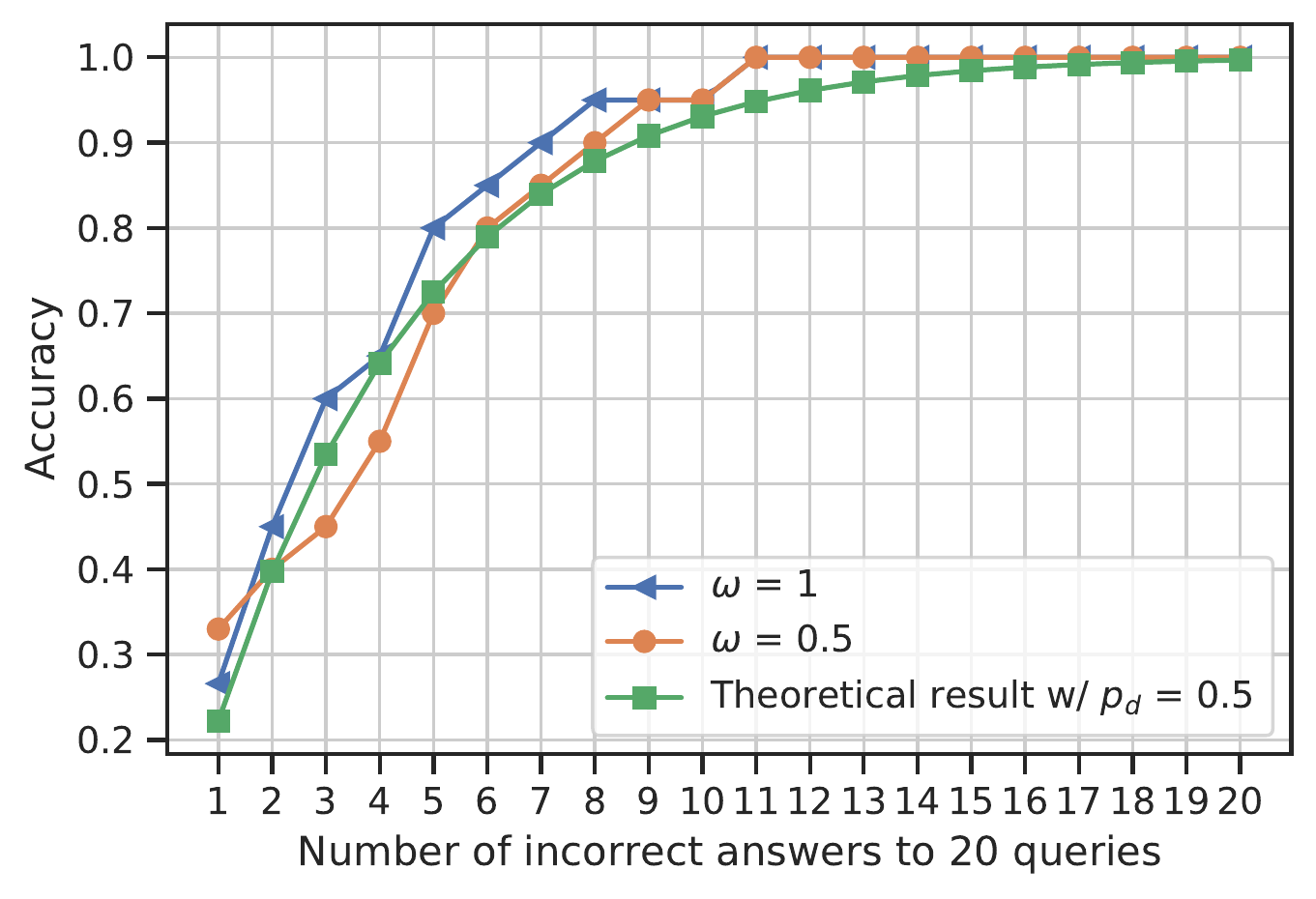}
\caption{Accuracy of the cheating detection against adding data records strategy
\label{fig:test_anal_experiment_theoretial_scaleup}} 
\end{figure}

\descr{Impact of privacy budget.} 
We vary the total privacy budget of a participant $P_1$ for the server $\epsilon_{\mathcal{S}}$ and for the querying participant $P_2$, $\epsilon = \epsilon_{\mathcal{S}} = [0.1:1.0]$. Hence, the noise added to $P_1$'s answer is drawn from the Laplace distribution of scale $m/(\epsilon+\epsilon_{\mathcal{S}})=m/2\epsilon$ where $m$ is the number of all queries. 

For modifying data records strategy, we evaluate the impact of modifying rate $\alpha=0.1$.
Figure \ref{fig:impact_epsilon_modifying_data} indicates that for a given number of incorrect answers our cheating detection accuracy increases with higher privacy budget and achieves similar performance when $\epsilon = [0.5:1.0]$. For smaller privacy budgets $\epsilon = 0.1$ or $0.3$, as higher maximum noise is accepted, there is more chance that an answer from modified dataset after being added with noise falls within the acceptance range. Thus, there might be more FNs that results in lower accuracy.

\captionsetup[figure]{justification=centering}
\begin{figure}[h]
\centering
\includegraphics[width=0.7\linewidth]{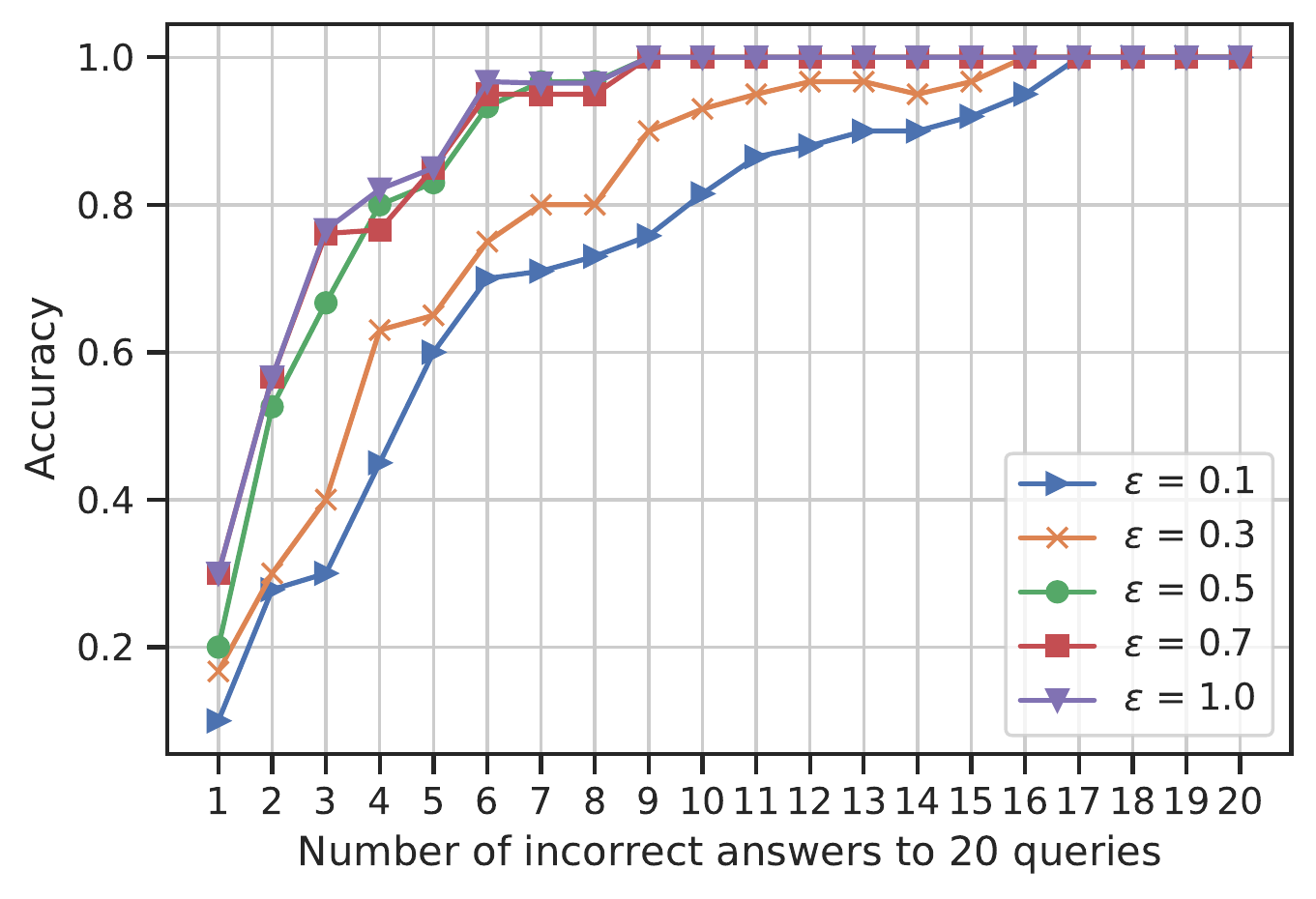}
\caption{Impact of privacy budget on the accuracy of cheating detection against modifying data records strategy ($\alpha = 0.1$) 
\label{fig:impact_epsilon_modifying_data}} 
\end{figure}

\captionsetup[figure]{justification=centering}
\begin{figure}[h]
\centering
\includegraphics[width=0.7\linewidth]{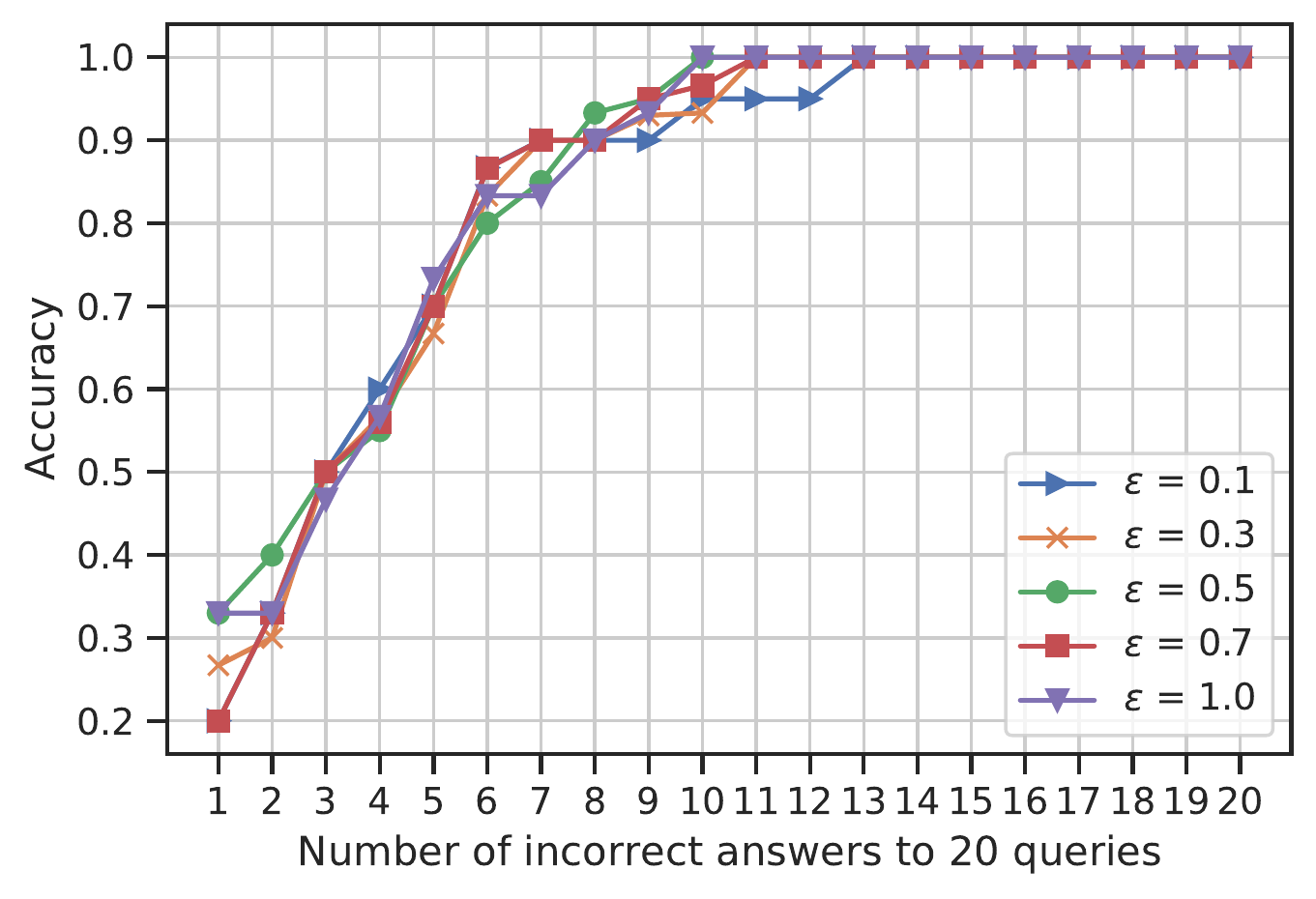}
\caption{Impact of privacy budget on the accuracy of cheating detection against adding data records strategy ($\omega = 0.5$)    \label{fig:impact_epsilon_injecting_data}}
\end{figure}

For adding fake data records strategy, we choose $\omega=0.5$ as the adding rate to generate a fake dataset $D^*$. Figure \ref{fig:impact_epsilon_injecting_data} shows that the cheating detection accuracy is not sensitive to privacy budget. This is because the added data records significantly scale up the test query answers. Adding noise to the answer hence makes it rather far from the acceptance range. Thus, privacy budgets do not  affect the cheating detection accuracy.

In summary, our cheating detection accuracy is only sensitive to small privacy budgets if there is a relatively small change to the true dataset and  is not sensitive to privacy budgets if a large change is made to the true dataset. 

\section{Performance Evaluation}
\label{sec:performance}

In this section, we evaluate the performance of DataRing, in terms of the computation and the communication overhead. We first evaluate DataRing's performance when the data domain cap increases. We then present the security-performance trade-off in the query evaluation phase and finally compare DataRing with the most relevant work of Drynx \citep{Froelicher2019} and Unlynx \citep{Froelicher2017}.

\subsection{Runtime}
We first evaluate the impact of the data domain cap $a$ on the runtime of DataRing.

\descr{Effect of Domain Cap on Runtime.}
Specifically, we vary the domain cap $a \in \{4, 6, 8, 10\}$ while keeping the default value for other parameters as in Table \ref{table:table:default_parameter}. 

Table \ref{fig:scale_up_factor_effect} shows the runtime of key operations in DataRing. The runtime of shuffling data domain, PV sampling and PV verification scales almost linearly with the domain cap $a$. The most expensive operation is the PV sampling as it requires randomly selecting $V$ labels and mapping all labels to their original order due to permutation. 
However, these operations are executed only once in the partial view phase. 
In the query evaluation phase, the runtime for computing real/test query's answer slightly increases with the increase of $a$. For instance, when $a$ is doubled from 4 to 8, the runtime to compute a real/test query's answer is increased by around 1 second. 
The runtime for formulating queries is slightly increased with larger domain cap $a$ as the server can pre-compute encryptions of 0s and 1s offline. The exception is formulating \textit{Test V} query as the server must re-randomise all ciphertexts in the obtained partial view. 
As we can see, even though higher domain caps mostly affects the runtime of operations in the partial view phase, these are one-off operations. In additions, pre-computing encryptions of 0s and 1s and using them to form queries can significantly reduce runtime of operations in query evaluation phase. Thus, it is feasible to use a relatively large domain cap to balance the trade-off between privacy and efficiency.

\begin{table*}[]
    \centering
    \captionof{table} {Effect of domain cap on average runtime of each operation. ($(*)$ is inquiring participant's operation) \label{fig:scale_up_factor_effect}}
    \resizebox{\columnwidth}{!}{
    \begin{tabular}{c|c c c c||c|c c c c}
        Participant's& & & Runtime (s) & & Servers'& & & Runtime (s) & \\
        Operation& a = 4& a = 6& a = 8& a = 10& Operation & a = 4& a = 6& a = 8& a = 10 \\
        \hline
        Shuffling domain& 9.35 &14.54  &19.56 &24.69 & Sampling PV& 25.60 &38.83  &48.92 &61.00  \\
        \hline
    Real Query Form.$^{(*)}$ & 1.63 &2.40  &3.12 &3.86  & Verif. PV  &0.56  &0.85  &1.13  &1.41  \\
        \hline
         Test L Ans.& 2.86 & 3.35 &3.87 &4.36 & Test L Form.& 1.63 & 2.44  &3.27 &3.92  \\
         \hline
         Test V Ans.& 2.91  & 3.52 &4.11 &4.45 & Test V Form.& 10.30 & 15.43 &20.62 &25.67  \\
         \hline
         Test N Ans.& 2.97  &3.57  & 4.17 & 4.51& Test N Form.& 2.44  & 3.66 &4.90 &6.14  \\
         \hline
        Real Query Ans.&2.92  &3.44  &3.98 &4.52  & Test Ans. Verif. &0.16  &0.24  &0.32 &0.40  \\
         
    \end{tabular}
    }
\end{table*}

\descr{Effect of Cheating Detection on Runtime.}
We consider two participants $P_i$ and $P_j$ who are assumed to have proceeded to the query evaluation phase after submitting honest partial views. 
We measure the total runtime for participant $P_i$ to evaluate 10 queries on participant $P_j$'s dataset consisting of 500,000 records with 10 attributes.
To detect cheating, we conduct 10 test queries during the query evaluation. Thus, $P_j$ must respond to 20 queries. Assuming that $P_j$ is a cheating party which uses a modified (fake) dataset that contains 15\% of dataset as fake data to evaluate 10 queries at random while using its true dataset to evaluate other 10 queries. To calculate the total runtime in the query evaluation phase, we omit the runtime for participant $P_i$ to form its real queries and the runtime for the server to generate test queries as these can be done off-line. Experimental result shows that by adding test queries to the process, the server successfully detects $P_j$ as cheating while introducing additional overhead to total runtime. Specifically, to evaluate 10 queries DataRing takes 90.63 seconds. Although without any test queries the total runtime could be reduced by 27\%, this does not ensure the correctness of participant's answers. 

\subsection{Communication Overhead}\label{subsec:com_overhead}
Our EC-ElGamal encryption scheme relies on prime256v1 curve with 128-bit security. Each ciphertext has a size of 66 bytes. 

In the partial view phase, $P$ sends set 
$\mathsf{lbs}$ to server $S_1$ and an inverse permutation vector $\sigma^{-1}$ to server $S_2$ that incurs an overhead of $aN \times 8$ bytes and $aN \times 4$ bytes overhead, respectively. $S_1$ then sends encrypted vector $\llbracket \mathsf{lbs} \rrbracket$ to $S_2$, where each element of $\llbracket \mathsf{lbs} \rrbracket$ includes an enumerated label and a ciphertext, that results in an overhead of $aN \times 70$ bytes.  
In the query evaluation phase, to evaluate a single query on $P$'s dataset, one server in $\mathcal{S}$ sends the encrypted query with a communication overhead of $aN \times 66$ bytes to $P$. Participant $P$ responds to $\mathcal{S}$ with an encrypted answer that is only one ciphertext of size $66$ bytes. 

Given parameters as in Table~\ref{table:table:default_parameter}, the overhead in the partial view phase includes 16MB and 8MB for $P$ to send set 
$\mathsf{lbs}$ and vector $\sigma^{-1}$ to $S_1$ and $S_2$, respectively, and 140MB for interaction between the two servers. 
Assuming that there are 10 real queries to be evaluated on $P$'s dataset, $\mathcal{S}$ will send these 10 real queries and 10 test queries to $P$ (in encrypted form). Thus, the overhead to execute 10 queries is $20 \times 132$MB $= 2640$MB at the server side and $20 \times 66$ bytes $ = 1.32$kB at the participant side. 

\subsection{Comparison with Existing Work}

We now compare DataRing's performance with Drynx \citep{Froelicher2019} and Unlynx \citep{Froelicher2017}, the two most relevant systems to our work. We adapt our evaluation to match some of their relevant settings for comparison. 

\begin{table}[b]
    \centering
    \captionof{table} {Comparison of DataRing's and Drynx's performance \label{fig:compare_drynx}}
    \begin{tabular}{c| c| c}
         & Drynx & DataRing\\
        \hline
        Runtime to  eval. & 2.6 s& 1.0 s\\
        a mean query & &\\
        \hline
        Comm. overhead  & 1.28kB (answers) & 1.32kB (answers)\\ 
        eval. 10 queries & 2.43MB (proofs) &52.8MB (tests)\\
    \end{tabular}
\end{table}

\descr{DataRing vs. Drynx.} Table \ref{fig:compare_drynx} shows a comparison between DataRing and Drynx performance in runtime and communication overhead. According to \citep[\S IX.B]{Froelicher2019}, Drynx takes 2.6 seconds to evaluate a mean query. 
In DataRing, we evaluate runtime of a query on the average income of borrowers in the Lending Club dataset of 20K records. To evaluate this mean query, a sum query and a count query is performed and the results are sent to the querying participant. Additional test queries are also conducted to ensure the correctness of the answer. 
The experimental result shows that DataRing takes only 1 second to evaluate a mean query.   
For communication overhead, Drynx reports in \citep[\S IX.B]{Froelicher2019} that the total bandwidth overhead for 10 participants communicating with the servers including 1.28kB for sending answers and 2.43MB for sending and storing input-range proofs. We consider these reports as being equivalent to the overhead incurs when a participant evaluates 10 queries on its dataset and communicates the answers to servers. 
In DataRing, to evaluate 10 queries, the communication overhead includes $1.32$kB for the participant sending its answers of real queries to the server and 52.8MB for the participant and the server communicating 10 test queries and answers during cheating detection. 
Although Drynx poses less communication overhead than DataRing, it is noted that Drynx does not check the correctness of data inputs to the queries. Drynx can only guarantee that the data inputs are within  acceptable semantic bounds, e.g., age between 0 and 120. 

\descr{DataRing vs. Unlynx.} We also compare DataRing's performance with Unlynx's private survey application \citep[\S 8]{Froelicher2017} in which it considers honest-but-curious and non-colluding servers. 
We evaluate the same query form that was used in Unlynx. Specifically, we measure the total runtime for evaluating the average income of borrowers (who own a house and have grade in $[0:6]$, and be grouped by loan term (36 months and 60 months) from 20 participants each has a dataset of 20K borrowers. These datasets are extracted from the Lending Club dataset. We simply conduct a count and a sum query to each participant with regard to each loan term.
The experimental results are showed in Table \ref{fig:compare_unlynx}.
DataRing takes slightly more time than Unlynx to evaluate this query as it ensures the correctness of answers by conducting cheating detection. Unlynx does not guarantee the correctness of query answer. 
As Unlynx does not report its communication overhead in this evaluation, we omit this comparison. 

\begin{table}[ht]
    \centering
    \captionof{table} {Comparison of DataRing's and Unlynx's performance \label{fig:compare_unlynx}}
    \begin{tabular}{c| c| c}
         Runtime & Unlynx & DataRing\\
        \hline
        Eval. mean query from 20 parties & 151 s& 161 s\\
    \end{tabular}
\end{table}

\section{Related work}
\label{sec:related_work}
Most privacy-preserving data sharing systems assume honest or honest-but-curious data owners.
Several works leverage homomorphic encryption and zero-knowledge proofs to ensure participants' data security and computational correctness \citep{Kim2018,Froelicher2017,SMCQL_2017,Pilatus_ECEG}. For instance, in Unlynx \citep{Froelicher2017}, data owners supply their query's response in encrypted form to a group of servers that aggregate the query's end result and provide proofs of the correctness of their computation. \cite{SMCQL_2017} proposed SMCQL which enables queries on datasets hosted by multiple data owners without revealing sensitive information of individuals in each dataset.  SMCQL assumes that data owners are trusted to faithfully execute the protocol. 
 
\cite{froelicher2020scalable} proposed a system for privacy-preserving distributed learning that protects the data and model confidentiality using homomorphic encryption. 
Generally, the main focus of these works is to protect data confidentiality and privacy, with less concern about the correctness of data being shared.

A few recent works in privacy-preserving data sharing have circumnavigated this hurdle by assuming that some background or public information about a party's dataset is known in advance, which limits the extent to which parties can deviate from the analysis on their true inputs. The Drynx system~\citep{Froelicher2019}  checks if the input attributes of a dataset are within a certain range (in a privacy-preserving manner), where the range is determined by the size of the parties dataset and semantic knowledge of these attributes, e.g., age can only be within the range 0 and 150. Arguably, this only discards invalid inputs, and participants may still deviate from their true inputs by picking arbitrary values within these ranges. Likewise, Prio~\citep{corrigan2017prio}, a system for privately computing aggregate statistics, employs a public predicate to ensure the inputs are within publicly allowed ranges. Apart from this, Prio also assumes the much simpler setting where each participant has a single datum instead of datasets with a large number of records. Helen, a system proposed by~\cite{helen2019} ensures the correctness of the collaboratively machine learning model by directing the participants to broadcast an encrypted summary of their datasets in advance. These summaries are used as the ground truth to verify the correctness of the learned model. However, the correctness of the summaries themselves is not verified. In contrast, DataRing ensures the correctness of data inputs and query results by requesting participants to provide a partial view (PV) of their datasets, verifying the trustworthiness of the PV, and conducting hidden tests to ensure that the participants continue to use datasets that are consistent with their partial views while answering queries.

\section{Conclusion}
\label{sec:conclusion}
We have proposed DataRing, a system that enables privacy-preserving data sharing among mutually mistrusting data owners while ensuring the correctness of query answers. DataRing ensures confidentiality of participant's dataset and privacy of individuals in the dataset by using cryptography primitives and differential privacy. It guarantees correctness of data inputs and query answers through the collection of partial view of participant's dataset and cheating detection during the query evaluation. 

\ifjpc
\else
\section*{Acknowledgments}
This work was conducted with funding received under the Pilot Grants scheme from New South Wales (NSW) Cyber Security Network for the Data Ring project.
\fi



\bibliography{references}

\begin{thebibliography}{35}
\providecommand{\natexlab}[1]{#1}
\providecommand{\url}[1]{\texttt{#1}}
\expandafter\ifx\csname urlstyle\endcsname\relax
  \providecommand{\doi}[1]{doi: #1}\else
  \providecommand{\doi}{doi: \begingroup \urlstyle{rm}\Url}\fi

\bibitem[ama()]{amazon_ec2}
Amazon ec2.
\newblock \url{https://aws.amazon.com/ec2/}.
\newblock Accessed: 2020-06-11.

\bibitem[ope()]{openssl_prime256}
Openssl: Elliptic curve cryptography parameters.
\newblock
  \url{https://docs.huihoo.com/doxygen/openssl/1.0.1c/ecparam_8c_source.html}.
\newblock Accessed: 2020-06-11.

\bibitem[rfc()]{rfc_prime256_secp256r1}
Elliptic curve cryptography subject public key information.
\newblock \url{https://www.ietf.org/rfc/rfc5480.txt}.
\newblock Accessed: 2020-06-11.

\bibitem[Acar et~al.(2018)Acar, Aksu, Uluagac, and Conti]{phe_survey}
A.~Acar, H.~Aksu, A.~S. Uluagac, and M.~Conti.
\newblock A survey on homomorphic encryption schemes: Theory and
  implementation.
\newblock \emph{ACM Comput. Surv.}, 51\penalty0 (4), July 2018.

\bibitem[Aguilar~Melchor et~al.(2016)Aguilar~Melchor, Barrier, Fousse, and
  Killijian]{xpir}
C.~Aguilar~Melchor, J.~Barrier, L.~Fousse, and M.-O. Killijian.
\newblock {XPIR : Private Information Retrieval for Everyone}.
\newblock \emph{{Proceedings on Privacy Enhancing Technologies}},
  2016:\penalty0 155--174, 2016.
\newblock URL \url{https://hal.archives-ouvertes.fr/hal-01396142}.

\bibitem[Bater et~al.(2017)Bater, Elliott, Eggen, Goel, Kho, and
  Rogers]{SMCQL_2017}
J.~Bater, G.~Elliott, C.~Eggen, S.~Goel, A.~Kho, and J.~Rogers.
\newblock Smcql: Secure querying for federated databases.
\newblock \emph{Proc. VLDB Endow.}, 10\penalty0 (6):\penalty0 673–684, Feb.
  2017.

\bibitem[Bindschaedler et~al.(2017)Bindschaedler, Shokri, and
  Gunter]{Vincent2017}
V.~Bindschaedler, R.~Shokri, and C.~A. Gunter.
\newblock Plausible deniability for privacy-preserving data synthesis.
\newblock \emph{Proc. VLDB Endow.}, 10\penalty0 (5):\penalty0 481–492, Jan.
  2017.

\bibitem[Burkhalter and Shafagh(2018)]{elgamal}
L.~Burkhalter and H.~Shafagh.
\newblock Additive homomorphic ec-elgamal.
\newblock \url{https://github.com/lubux/ecelgamal}, 2018.
\newblock URL \url{https://github.com/lubux/ecelgamal}.
\newblock Accessed: Oct 26, 2019.

\bibitem[{Chor} et~al.(1995){Chor}, {Goldreich}, {Kushilevitz}, and
  {Sudan}]{PIR_1995}
B.~{Chor}, O.~{Goldreich}, E.~{Kushilevitz}, and M.~{Sudan}.
\newblock Private information retrieval.
\newblock In \emph{Proceedings of IEEE 36th Annual Foundations of Computer
  Science}, pages 41--50, Oct 1995.

\bibitem[Corrigan-Gibbs and Boneh(2017)]{corrigan2017prio}
H.~Corrigan-Gibbs and D.~Boneh.
\newblock Prio: Private, robust, and scalable computation of aggregate
  statistics.
\newblock In \emph{14th USENIX Symposium on Networked Systems Design and
  Implementation (NSDI 17)}, pages 259--282, 2017.

\bibitem[Dwork(2006)]{Dwork2006}
C.~Dwork.
\newblock Differential privacy.
\newblock In \emph{Proceedings of the 33rd international conference on
  Automata, Languages and Programming-Volume Part II}, pages 1--12.
  Springer-Verlag, 2006.

\bibitem[Dwork et~al.(2006)Dwork, McSherry, Nissim, and Smith]{calib-noise}
C.~Dwork, F.~McSherry, K.~Nissim, and A.~Smith.
\newblock Calibrating noise to sensitivity in private data analysis.
\newblock In \emph{Theory of cryptography conference}, pages 265--284.
  Springer, 2006.

\bibitem[Dwork et~al.(2014)Dwork, Roth, et~al.]{Dwork2014}
C.~Dwork, A.~Roth, et~al.
\newblock The algorithmic foundations of differential privacy.
\newblock \emph{Foundations and Trends in Theoretical Computer Science},
  9\penalty0 (3-4):\penalty0 211--407, 2014.

\bibitem[Froelicher et~al.(2017)Froelicher, Egger, Sousa, Raisaro, Huang,
  Mouchet, Ford, and Hubaux]{Froelicher2017}
D.~Froelicher, P.~Egger, J.~S. Sousa, J.~L. Raisaro, Z.~Huang, C.~Mouchet,
  B.~Ford, and J.-P. Hubaux.
\newblock Unlynx: a decentralized system for privacy-conscious data sharing.
\newblock \emph{Proceedings on Privacy Enhancing Technologies}, 2017\penalty0
  (4):\penalty0 232--250, 2017.

\bibitem[Froelicher et~al.(2020)Froelicher, Troncoso-Pastoriza, Pyrgelis, Sav,
  Sousa, Bossuat, and Hubaux]{froelicher2020scalable}
D.~Froelicher, J.~R. Troncoso-Pastoriza, A.~Pyrgelis, S.~Sav, J.~S. Sousa,
  J.-P. Bossuat, and J.-P. Hubaux.
\newblock Scalable privacy-preserving distributed learning, 2020.

\bibitem[{Froelicher} et~al.(2020){Froelicher}, {Troncoso-Pastoriza}, {Sousa},
  and {Hubaux}]{Froelicher2019}
D.~{Froelicher}, J.~R. {Troncoso-Pastoriza}, J.~S. {Sousa}, and J.~{Hubaux}.
\newblock Drynx: Decentralized, secure, verifiable system for statistical
  queries and machine learning on distributed datasets.
\newblock \emph{IEEE Transactions on Information Forensics and Security},
  15:\penalty0 3035--3050, 2020.

\bibitem[Gallian(2010)]{con-aa}
J.~Gallian.
\newblock \emph{Contemporary abstract algebra}.
\newblock Nelson Education, 7 edition, 2010.

\bibitem[Gertner et~al.(2000)Gertner, Ishai, Kushilevitz, and Malkin]{SPIR}
Y.~Gertner, Y.~Ishai, E.~Kushilevitz, and T.~Malkin.
\newblock Protecting data privacy in private information retrieval schemes.
\newblock \emph{Journal of Computer and System Sciences}, 60\penalty0
  (3):\penalty0 592–629, June 2000.

\bibitem[Goldreich(2009)]{foc-2}
O.~Goldreich.
\newblock \emph{Foundations of cryptography: volume 2, basic applications}.
\newblock Cambridge university press, 2009.

\bibitem[Hankerson et~al.(2010)Hankerson, Menezes, and Vanstone]{guideToECC}
D.~Hankerson, A.~J. Menezes, and S.~Vanstone.
\newblock \emph{Guide to Elliptic Curve Cryptography}.
\newblock Springer Publishing Company, Incorporated, 1st edition, 2010.
\newblock ISBN 1441929290.

\bibitem[Hunt et~al.(2018)Hunt, Song, Shokri, Shmatikov, and
  Witchel]{hunt2018chiron}
T.~Hunt, C.~Song, R.~Shokri, V.~Shmatikov, and E.~Witchel.
\newblock Chiron: Privacy-preserving machine learning as a service.
\newblock \emph{arXiv preprint arXiv:1803.05961}, 2018.

\bibitem[Hynes et~al.(2018)Hynes, Cheng, and Song]{hynes-tee}
N.~Hynes, R.~Cheng, and D.~Song.
\newblock Efficient deep learning on multi-source private data.
\newblock \emph{arXiv preprint arXiv:1807.06689}, 2018.

\bibitem[{Jagielski} et~al.(2018){Jagielski}, {Oprea}, {Biggio}, {Liu},
  {Nita-Rotaru}, and {Li}]{poisoning_attack}
M.~{Jagielski}, A.~{Oprea}, B.~{Biggio}, C.~{Liu}, C.~{Nita-Rotaru}, and
  B.~{Li}.
\newblock Manipulating machine learning: Poisoning attacks and countermeasures
  for regression learning.
\newblock In \emph{2018 IEEE Symposium on Security and Privacy (SP)}, pages
  19--35, 2018.

\bibitem[Kan(2019)]{LendingClub}
W.~Kan.
\newblock Lending club loan data: Analyze lending club's issued loans.
\newblock \url{https://www.kaggle.com/wendykan/lending-club-loan-data}, 2019.
\newblock URL \url{https://www.kaggle.com/wendykan/lending-club-loan-data}.
\newblock Accessed: Dec 10, 2019.

\bibitem[Kim et~al.(2018)Kim, Song, Wang, Xia, and Jiang]{Kim2018}
M.~Kim, Y.~Song, S.~Wang, Y.~Xia, and X.~Jiang.
\newblock Secure logistic regression based on homomorphic encryption: Design
  and evaluation.
\newblock \emph{JMIR medical informatics}, 6\penalty0 (2):\penalty0 e19, 2018.

\bibitem[Koblitz(1987)]{koblitz1987elliptic}
N.~Koblitz.
\newblock Elliptic curve cryptosystems.
\newblock \emph{Mathematics of computation}, 48\penalty0 (177):\penalty0
  203--209, 1987.

\bibitem[Lindell(2017)]{sim-tut}
Y.~Lindell.
\newblock How to simulate it--a tutorial on the simulation proof technique.
\newblock In \emph{Tutorials on the Foundations of Cryptography}, pages
  277--346. Springer, 2017.

\bibitem[Ohrimenko et~al.(2016)Ohrimenko, Schuster, Fournet, Mehta, Nowozin,
  Vaswani, and Costa]{ohrimenko2016oblivious}
O.~Ohrimenko, F.~Schuster, C.~Fournet, A.~Mehta, S.~Nowozin, K.~Vaswani, and
  M.~Costa.
\newblock Oblivious multi-party machine learning on trusted processors.
\newblock In \emph{25th USENIX Security Symposium Security 16)}, pages
  619--636, 2016.

\bibitem[Rice(2006)]{rice2006mathematical}
J.~A. Rice.
\newblock \emph{Mathematical statistics and data analysis}.
\newblock Cengage Learning, 2006.

\bibitem[Shafagh et~al.(2015)Shafagh, Hithnawi, Droescher, Duquennoy, and
  Hu]{Talos_ECEG}
H.~Shafagh, A.~Hithnawi, A.~Droescher, S.~Duquennoy, and W.~Hu.
\newblock Talos: Encrypted query processing for the internet of things.
\newblock In \emph{Proceedings of the 13th ACM Conference on Embedded Networked
  Sensor Systems}, SenSys ’15, page 197–210, 2015.

\bibitem[Shafagh et~al.(2017)Shafagh, Hithnawi, Burkhalter, Fischli, and
  Duquennoy]{Pilatus_ECEG}
H.~Shafagh, A.~Hithnawi, L.~Burkhalter, P.~Fischli, and S.~Duquennoy.
\newblock Secure sharing of partially homomorphic encrypted iot data.
\newblock In \emph{Proceedings of the 15th ACM Conference on Embedded Network
  Sensor Systems}, SenSys ’17, 2017.

\bibitem[Shoup(2009)]{shoup-nt}
V.~Shoup.
\newblock \emph{A computational introduction to number theory and algebra}.
\newblock Cambridge university press, 2009.

\bibitem[Shuster(2014)]{hypergeo_book}
J.~J. Shuster.
\newblock \emph{Hypergeometric Distribution: Introduction}.
\newblock American Cancer Society, 2014.

\bibitem[Yekhanin(2010)]{pir_sergey}
S.~Yekhanin.
\newblock Private information retrieval.
\newblock \emph{Commun. ACM}, 53\penalty0 (4):\penalty0 68–73, Apr. 2010.

\bibitem[{Zheng} et~al.(2019){Zheng}, {Popa}, {Gonzalez}, and
  {Stoica}]{helen2019}
W.~{Zheng}, R.~A. {Popa}, J.~E. {Gonzalez}, and I.~{Stoica}.
\newblock Helen: Maliciously secure coopetitive learning for linear models.
\newblock In \emph{2019 IEEE Symposium on Security and Privacy (SP)}, pages
  724--738, May 2019.

\end{thebibliography}
\bibliographystyle{abbrvnat}

\appendix

\section{EC-ElGamal Cryptosystem}
\label{app:ec-elgamal}
The following is a well known result for cyclic groups. 
\begin{thm}
\label{the:lag}
Let $\mathcal{E}$ be an elliptic curve over a finite field $\mathbb{F}_p$ and let $B$, a point on $\mathcal{E}(\mathbb{F}_p)$, be a generator of the cyclic subgroup of prime order $q$. Then for every $j \in [1, q-1]$, $jB$ is also a generator of the same group. 
\end{thm}
\begin{proof}
This is a corollary of Lagrange's theorem~\citep[\S 7]{con-aa}.
\end{proof}

\begin{thm}
Let $q$ be a prime, and let $r \in \mathbb{Z}_q^* = \{1, 2, \ldots, q-1\}$. Then the function $f_r: \mathbb{Z}_q^* \rightarrow \mathbb{Z}_q^*$, defined as $f_r(\alpha) = \alpha r \; (\bmod \; q)$ is a one-to-one map.
\end{thm}
\begin{proof}
See Theorem 2.5 in~\citep[\S 2.3]{shoup-nt}.
\end{proof}

\begin{cor}
\label{cor:uni-alpha}
Let $q$ be a prime, let $r \in \mathbb{Z}_q^*$ and let $\alpha$ be drawn uniformly at random from $\mathbb{Z}_q^*$. Then $\alpha r \; (\bmod \; q)$ is uniformly distributed over $\mathbb{Z}_q^*$.
\end{cor}

The following provides details of basic algorithms in EC-ElGamal cryptosystem used in DataRing.

\descr{Key generation.}
Let $k$ denote the private key which is sampled uniformly at random from $[1, q-1]$. The public key is $K = kB$.

\descr{Encryption.}
Let $X = xB$ be the mapping of a message $x$ as a point on the curve. The encryption of $X$ is the tuple $(C_1, C_2)= (rB, X + rK)$, where $r$ is sampled uniformly at random from $[1, m-1]$, and $K$ is the public-key of the receiver. We denote this by $\mathsf{Enc}_K(x; r)$.

\descr{Decryption.}
Decryption of the ciphertext $\mathsf{Enc}_K(x; \cdot) = (C_1, C_2)$, is defined as $X = C_2 - kC_1$, where $k$ is the private-key of the receiver. The plaintext $x$ is then extracted from $X$. We denote this operation by $\mathsf{Dec}_k$.

\descr{Scalar multiplication of ciphertext.}
Given a scalar $\alpha \in [1, q-1]$, its multiplication with the ciphertext $\mathsf{Enc}_K(x; r) = (C_1, C_2) = (rB, X + rK)$ is defined as: $\mathsf{Enc}_K(\alpha x; \alpha r) = (\alpha C_1, \alpha C_2) = (\alpha rB, \alpha (xB + rK))$. This can be implemented using point addition and doubling operations~\citep[\S 3.3]{guideToECC}.

\descr{Re-encryptions of zero.} A corollary of the scalar multiplication property is that encryptions of $x = 0$, i.e., $\mathsf{Enc}_K(0; r) = (C_1, C_2) = (rB, 0B + rK) = (rB, rK)$, for some $r \in [1, q-1]$ can be re-randomized by choosing a random $\alpha \in [1, q-1]$ and updating the ciphertext as before, resulting in $\mathsf{Enc}_K(0; \alpha r)$. Note that $\alpha r$ modulo $q$ is uniformly distributed over $[1,q-1]$ (Corollary~\ref{cor:uni-alpha}).

\section{Security Proofs}
\label{app:security}

\subsection{Security: Partial View Collection}
\label{app:pvc}
We prove the security of the partial view collection phase using the simulation paradigm~\citep{sim-tut}. Let $P$ be any of the participants in $\mathcal{P}$.
Let $F$ denote the functionality for collecting the partial view in the ideal setting, and let $\Pi$ denote our protocol for collecting the partial view in the real-world setting. Let $z$ be the security parameter, used to generate keys in EC-ElGamal. Let $I_{P}$, $I_{S_1}$ and $I_{S_2}$ denote the inputs of $P$, $S_1$ and $S_2$, respectively. Let $\mathcal{A}$ denote the adversary in the ideal setting, which can control $P$, \emph{and/or} one of $S_1$ and $S_2$ (the latter two in an honest-but-curious way). Let $\mathcal{B}$ denote the adversary that controls $P$ in the real-setting. Further, let $\mathcal{C}$ denote the adversary that controls $S_1$ or $S_2$.

\descr{Ideal Execution.} In the ideal execution of $F$, we assume a trusted third party (TTP). The common inputs of parties $\mathcal{P}$, $S_1$ and $S_2$ are: the domain $\mathcal{D}$ (and its size), $P$'s dataset size $N$, and the size of the partial view $V$. In addition, $P$'s input includes the dataset $D$. The servers $S_1$ and $S_2$ also receive the set $\mathcal{L}$, i.e., the background knowledge of $S$.
All honest parties hand over their inputs to TTP. The corrupted party $P$, controlled by $\mathcal{A}$, may send a dataset $D'$ different from $D$. The TTP upon receiving $D'$ (possibly equal to $D$), first checks if each row of $D'$ is unique (i.e., each point of $D'$ has cardinality 1). If not, it terminates the ideal execution, and sends the message ``invalid dataset'' to $S_1$ and $S_2$. Otherwise, the TTP samples a random sample (partial view) of size $V$ from $D'$. Upon receiving a query on a label $l \in \mathcal{L}$ from $S_1$ and $S_2$, it sends $1$ to both if the corresponding point is in the partial view, and $0$ otherwise. Apart from this, $P$, $S_1$ and $S_2$ receive no further output from TTP. In particular, party $P$, receives no output. Let $\textsc{Ideal}_{F, \mathcal{A}}(I_{P}, I_{S_1}, I_{S_2}, z)$ denote the outputs of the parties and the adversary in the ideal setting.

\descr{Real Execution.} In the real model, the protocol $\Pi$ described in Section~\ref{subsec:pvc}, describes the functionality for collecting the partial view. Let $\textsc{Real}_{\Pi, \mathcal{B}, \mathcal{C}}(I_{P}, I_{S_1}, I_{S_2}, z)$ denote the outputs of the parties and the adversaries in the real setting. 

\begin{defi}
\label{def:real-ideal}
Let $F$ be the functionality for collecting the partial view, and let $\Pi$ be a protocol that computes $F$. We say that $\Pi$ securely computes $F$, if for every pair of probabilistic polynomial time adversaries $(\mathcal{B}(z), \mathcal{C}(z))$ in the real model, there exists a probabilistic polynomial time adversary $\mathcal{A}(z)$ in the ideal model, such that 
\[
\textsc{Ideal}_{F, \mathcal{A}}(I_{P}, I_{S_1}, I_{S_2}, z) \approx_c \textsc{Real}_{\Pi, \mathcal{B}, \mathcal{C}}(I_{P}, I_{S_1}, I_{S_2}, z),
\]
where $\approx_c$ denotes computational indistinguishability, and we assume the adversaries to have auxiliary inputs. 
\end{defi}

Our main result is as follows.

\begin{thm}
If EC-ELGamal is semantically secure under the decisional Diffie Hellman (DDH) assumption, our protocol in Section~\ref{subsec:pvc} securely collects the partial view.
\end{thm}
\begin{proof}
We first separately consider the adversaries corrupting each party.
 
\descr{Corrupted Party $P$:} In the ideal model, the corrupted party sends the dataset $D'$, possibly different from $D$, to TTP. In the real model, the corrupted party can choose $D'$, and then send $\mathsf{lbs} = \{(\sigma(i), f_i)\}_{i \in \mathcal{D}}$ to $S_1$ and $\sigma^{-1}$ to $S_2$. Due to Theorem~\ref{the:fakeD}, we assume that $\sigma^{-1}$ is indeed the correct inverse permutation (and hence the adversary $\mathcal{B}$ does not corrupt this output). The (simulating) adversary $\mathcal{A}$, computes $\{\sigma^{-1}(\sigma(i)), f_i)\}_{i \in \mathcal{D}}$. It then constructs the dataset $D'$ in which the label $\sigma^{-1}(\sigma(i))$ is set to $f_i$, and sends $D'$ to TTP.

\descr{Corrupted Party $S_1$:} For simplicity, we assume there is only one server in $S_1$. The case of more than one servers is analogous, although requires a little more detail. In the ideal model, the input from this server (apart from the common inputs) is the set $\mathcal{L}$ of background information. The simulating adversary $\mathcal{A}$, first uses the security parameter $z$ to construct an EC-ElGamal private-public key pair $(k_1, K_1)$. It gives $(k_1, K_1)$ to $S_1$. The adversary $\mathcal{A}$ creates another private-public key pair $(k_2, K_2)$ using the security parameter $z$, and sends $K_S = K_1 + K_2$ to $S_1$, as the collective public key. Adversary
$\mathcal{A}$ constructs $\mathsf{lbs} = \{(i, f_i)\}_{i \in \mathcal{D}}$, where a random $N$ of the $f_i$'s
are 1, and the rest are $0$, and sends it to $S_1$. Upon receiving $\llbracket \mathsf{lbs} \rrbracket$ from $S_1$, it re-randomizes each pair $(i, \llbracket f \rrbracket)$ with  $(i, \llbracket f \rrbracket + \llbracket 0 \rrbracket)$, where $\llbracket 0 \rrbracket$ represents fresh encryptions of $0$ under $K_S$. It sets this as PV. For each $l \in \mathcal{L}$, the adversary $\mathcal{A}$ queries TTP, and receives bit $b_l$. It then replaces the corresponding entry in PV with $(l, \llbracket b_l \rrbracket)$. It then sends this PV to $S_1$. 
Upon receiving $l \in \mathcal{L}$ from $S_1$ (indicating the start of a threshold decryption operation from $S_1$), it does as follows. If $S_1$ is the last party in the threshold decryption, then it simply sends $S_1$ the partial decryption of the ciphertext under $l$ in PV using $k_2$. Otherwise, it receives the partial decryption of the ciphertext under $l$ in PV from $S_1$, discards it, and sends $b_l$ to $S_1$. 

\descr{Corrupted Party $S_2$:} Again, for simplicity, we assume there is only one server in $S_2$. Also, in the ideal model, the input from this server (apart from the common inputs) is the set $\mathcal{L}$ of background information. The simulating adversary $\mathcal{A}$, first uses the security parameter $z$ to construct an EC-ElGamal private-public key pair $(k_2, K_2)$. It gives $(k_2, K_2)$ to $S_2$. The adversary $\mathcal{A}$ creates another private-public key pair $(k_1, K_1)$ using the security parameter $z$, and sends $K_S = K_1 + K_2$ to $S_2$, as the collective public key. Adversary $\mathcal{A}$ creates a random permutation $\sigma$ and its inverse permutation $\sigma^{-1}$, and gives it to $S_2$. Adversary $\mathcal{A}$ constructs $\llbracket \mathsf{lbs} \rrbracket = \{(i, \llbracket 0 \rrbracket)\}_{i \in \mathcal{D}}$. For each $l \in \mathcal{L}$, the adversary $\mathcal{A}$ queries TTP, and receives bit $b_l$. It then replaces the corresponding entry in $\llbracket \mathsf{lbs} \rrbracket$ with $(l, \llbracket b_l \rrbracket)$. It again updates $\llbracket \mathsf{lbs} \rrbracket$ by applying $\sigma$ to the labels, and sends $\llbracket \mathsf{lbs} \rrbracket$ to $S_2$. Upon receiving $l \in \mathcal{L}$ from $S_2$ (indicating the start of a threshold decryption operation from $S_2$), it does as follows. If $S_2$ is the last party in the threshold decryption, then it simply sends $S_2$ the partial decryption of the ciphertext under $l$ in PV using $k_1$. Otherwise, it receives the partial decryption of the ciphertext under $l$ in PV from $S_2$, discards it, and sends $b_l$ to $S_2$. 

\descr{Corrupted Parties $P$ and $S_1$:} In this case, the simulating adversary $\mathcal{A}$ simply uses $\mathsf{lbs} = \{(\sigma(i), f_i)\}_{i \in \mathcal{D}}$ received from the corrupted party $P$ to construct $D'$. It sends $D'$ to TTP. If it receives ``invalid dataset'' from TTP, it sends the message to $S_1$. Otherwise, the simulation proceeds as in the case of only $S_1$ being corrupted.

\descr{Corrupted Parties $P$ and $S_2$:} In this case, the simulating adversary $\mathcal{A}$ again uses $\mathsf{lbs} = \{(\sigma(i), f_i)\}_{i \in \mathcal{D}}$ received from the corrupted party $P$ to construct $D'$. It sends $D'$ to TTP. If it receives ``invalid dataset'' from TTP, it sends the message to $S_2$. Otherwise, it sends $\sigma^{-1}$ to $S_2$, and proceeds as in the case of only $S_2$ being corrupted. 

It follows that our protocol securely collects the partial view if EC-ElGamal is semantically secure under the DDH assumption.
\end{proof}

\subsection{Security: Indistinguishability of Real and Test Queries}
\label{app:tvsr}

We ``parameterize'' the semantic security game by $m_\mathsf{q}/m$. This means with probability $m_\mathsf{q}/m$ the challenger sends the encryption of $m_0$, otherwise, it sends the encryption of $m_1$ to the adversary. Note that we assume $1 \le m_\mathsf{t} \le m_\mathsf{q}$. Thus, if $m_\mathsf{t} = m_\mathsf{q}$ then we retrieve the standard semantic security game. In our query indistinguishability game, the adversary is given a set of queries $\mathcal{Q} \cup \mathcal{T}$, where $\mathcal{Q}$ is a set of real queries and $\mathcal{T}$ is a set of test queries, and the adversary is told which query belongs to which set. The challenger picks $m_\mathsf{q}$ real and $m_\mathsf{t}$ test queries (not necessarily unique). The challenger then picks a random query from this set, encrypts it using EC-ElGamal, and sends it to the adversary. The adversary outputs ``real'' or ``test.'' The use of the security parameter $z$ (for ElGamal encryption), and the domain $\mathcal{D}$ is implicit in this game. 

Our reduction is as follows, our semantic security adversary $\mathcal{A}$ chooses two messages $m_0 = 0$ and $m_1 = 1$. It then constructs a real query as the $|\mathcal{D}|$-element vector all whose elements are 0. It then constructs a test query as the $|\mathcal{D}|$-element vector whose first element is 1, and the rest are all zeroes. Adversary $\mathcal{A}$ gives the two as the set of real and test queries to the query indistinguishability adversary $\mathcal{B}$. Upon submitting $m_0 = 0$ and $m_1 = 1$ to the challenger, $\mathcal{A}$ receives $\llbracket m_b \rrbracket$. The adversary $\mathcal{A}$ constructs the $|\mathcal{D}|$-element query vector:
\[
\llbracket Q \rrbracket = \begin{pmatrix} \llbracket m_b \rrbracket & \llbracket 0 \rrbracket & \llbracket 0 \rrbracket & \cdots & \llbracket 0 \rrbracket \end{pmatrix},
\]
where $\llbracket 0 \rrbracket$ are the encryptions of zero under the given public key $K$ of the EC-ElGamal cryptosystem. Adversary $\mathcal{A}$ gives $\llbracket Q \rrbracket$ to $\mathcal{B}$. If $\mathcal{B}$ returns ``real'', $\mathcal{A}$ outputs 0, else it outputs 1. Clearly the advantage exactly translates to the advantage in the semantic security game, as the probability that $\mathcal{B}$ receives a real or a test query is exactly $m_\mathsf{q}/m$, which is $1/2$, if $m_\mathsf{t} = m_\mathsf{q}$.


\section{Differential Privacy}
\label{DP_appendix}

Differential privacy~\citep{calib-noise}, is a definition of privacy tailored to statistical analysis of datasets. Informally, a mechanism (algorithm) satisfying the definition of differential privacy inherits the guarantee that the probability of any output of the algorithm with or without any single record in the dataset remains similar. 
More formally, given two neighbouring datasets $D_1$ and $D_2$ from a public domain $\mathcal{D}$, an algorithm $\mathcal{M}$ satisfies $\epsilon$-differential privacy, if for all subsets of outputs $O$ in the output range of $\mathcal{M}$ it holds that 
\[
\mathrm{P}(\mathcal{M}(D_1) \in O) \leq e^\epsilon \mathrm{P}(\mathcal{M}(D_2) \in O ).
\]

Here, neighbouring datasets means that $\lVert D_1 - D_2 \rVert_1 = 1$, i.e., the two datasets differ in one row. $\epsilon$ is the privacy parameter.

\descr{Composition Theorem of Differential Privacy.} Given algorithms $\mathcal{M}_1$, ..., $\mathcal{M}_k$ that satisfy $\epsilon_1$, \ldots, $\epsilon_k$-differential privacy, respectively, their combination defined by $\mathcal{M}_{[k]} = (\mathcal{M}_1,\ldots, \mathcal{M}_k)$ is an $\sum_{i=1}^{k}\epsilon_i$-differentially private algorithm \citep[\S 3.5]{Dwork2014}. 

Given parameter $\epsilon$ and query sensitivity $\Delta Q$, adding noise from the Laplace distribution of scale $\Delta Q/\epsilon$ to the answer $Q(D)$ is $\epsilon$-differentially private~\citep{Dwork2006}. 
If a set of $m$ count queries ($\Delta Q = 1$)  are evaluated over a dataset, to achieve $\epsilon$-differential privacy, according to the composition theorem of differential privacy, a noise drawn from the Laplace distribution of scale $m/\epsilon$ is added to each query answer.

\descr{Privacy and Utility Guarantee.}
To achieve $\epsilon$-differential privacy, given the query sensitivity $\Delta Q$, when evaluating query $Q$ on a dataset $D$ the Laplace mechanism outputs a query answer as $Q(D) + \upsilon$ where $\upsilon$ is drawn from a Laplace distribution with scale $\Delta Q/\epsilon$. The degradation of utility can be observed as the noise added to the query answer which guarantees $Pr[|\upsilon| \leq \frac{\Delta Q}{\epsilon}\ln{(1/\sigma)}] \geq 1- \sigma$, $\forall \sigma \in (0,1]$. This means with a probability of at least $1-\sigma$, the maximum degradation of utility is $\frac{\Delta Q}{\epsilon}\ln{(1/\sigma)}$, i.e., no query answer will be off more than an additive error of $\frac{\Delta Q}{\epsilon}\ln{(1/\sigma)}$ \citep[\S 3.3]{Dwork2014}.

\section{Notations}
\begin{table}[ht!]
\caption{Summary of frequently used notations}
\begin{tabular}{ l l } 
\hline
Notation &  Description \\
\hline
$S$ &  The set of servers\\
$P$ & Participant\\
$\mathcal{D}$ &   Data domain\\
$a$ & Domain cap\\
$N$ &  Dataset size\\
$V$ & Partial view size\\
$\rho$ & Ratio of $V$ to $N$\\
$\mathcal{L}$ & Servers' background knowledge\\
$L$ & Background knowledge size\\
$l$ & Label of a data record\\
$\mathsf{val}(l)$ & Count of the presence of a record with label $l$\\
$\eta$ & Tolerated false positive rate\\
$r_0$ & Threshold of known records to be found in partial view\\
$K_{S}$ & Collective public key\\
$\epsilon$ & Overall privacy budget\\
$\sigma$ & Random permutation\\
$\mathsf{lbs}$ & Set of labels and their flags in shuffled order\\
$\mathbf{u}$ & N-element binary vector\\
$\llbracket \mathbf{.} \rrbracket$ & Encrypted form under public key $K_{S}$\\
$P_c$ & Cheating participant\\
\hline
\end{tabular}
\end{table}

\end{document}